\documentclass[11pt]{article}
\renewcommand{\baselinestretch}{1}

\usepackage{amsmath,amssymb,amsfonts,amsthm, bbm}
\usepackage{epic,eepic,epsfig,longtable}
\usepackage{multirow,verbatim}
\usepackage{array}
\usepackage{epsfig}
\usepackage{setspace}
\usepackage{color}
\usepackage{pb-diagram}
\usepackage{fancyhdr}
\usepackage{appendix}
\usepackage{listings}
\usepackage{longtable}
\usepackage{lineno,hyperref}
\usepackage{subfigure}
\usepackage{mathrsfs}
\usepackage{stmaryrd}
\usepackage{appendix}
\usepackage{algorithmic}
\usepackage{algorithm}

\usepackage{graphicx}
\usepackage{enumerate}
\usepackage{natbib}
\setcitestyle{authoryear,open={(},close={)}}
\usepackage{url} % not crucial - just used below for the URL

\usepackage{esvect}

\usepackage{xr}
\providecommand{\norm}[1]{\left\lVert#1\right\rVert}

\numberwithin{equation}{section}

\renewcommand{\hat}{\widehat}

\renewcommand{\hat}{\widehat}

\newcommand{\bfm}[1]{\ensuremath{\mathbf{#1}}}

   \def\bD{\bfm D}  
\def\be{\bfm e}

   \def\bI{\bfm I}

   \def\bP{\bfm P}  \def\PP{\mathbb{P}}
     
     \def\RR{\mathbb{R}}

\def\bu{\bfm u}   \def\bU{\bfm U}  
\def\bv{\bfm v}   \def\bV{\bfm V}  
     
\def\bx{\bfm x}   \def\bX{\bfm X}  
\def\by{\bfm y}

 \def\cA{{\cal  A}}
 \def\cB{{\cal  B}}
 \def\cC{{\cal  C}}

 \def\cG{{\cal  G}}

 \def\cL{{\cal  L}}

 \def\cO{{\cal  O}}

 \def\cS{{\cal  S}}
 \def\cT{{\cal  T}}
 \def\cU{{\cal  U}}
 \def\cV{{\cal  V}}
 \def\cW{{\cal  W}}

\def\bzero{\bfm 0}

\newcommand{\bfsym}[1]{\ensuremath{\boldsymbol{#1}}}

\def\bbeta{\bfsym \beta}
\def\bgamma{\bfsym \gamma}

          \def\bepsilon{\bfsym \varepsilon}
             \def\bSigma{\bfsym \Sigma}
        \def\bLambda {\bfsym {\Lambda}}

\DeclareMathOperator*{\argmin}{argmin}

\DeclareMathOperator{\cov}{cov}

\DeclareMathOperator{\diag}{diag}

\DeclareMathOperator{\E}{E}

\DeclareMathOperator{\supp}{supp}
\DeclareMathOperator{\Var}{Var}

\DeclareMathOperator{\tr}{tr}

\DeclareMathOperator{\sign}{\rm sign}

\def\newpage{\vfill\eject}

\DeclareMathOperator*{\topabs}{\cT_{\mathrm{abs}}}

\newcommand{\beq}{\begin{equation}}
\newcommand{\eeq}{\end{equation}}
\newcommand{\beqn}{\begin{eqnarray}}
\newcommand{\eeqn}{\end{eqnarray}}
\newcommand{\beqnn}{\begin{eqnarray*}}
	\newcommand{\eeqnn}{\end{eqnarray*}}

\def\lasso{{\rm LASSO}}

\def\R{{\mathbb R}}
\def\N{{\mathbb N}}

\def\E{{\mathbb E}}

\def\P{{\mathbb P}}
\def\diag{{\rm diag}}

\def\sign{{\rm sign}}
\def\cov{{\rm cov}}

\def\supp{{\rm supp}}

\def\argmin{{\rm argmin}}

\numberwithin{equation}{section}
\newtheorem{thm}{Theorem}[section]

\newtheorem{lem}{Lemma}[section]
\newtheorem{cor}{Corollary}[section]
\newtheorem{prop}{Proposition}[section]
\newtheorem{ass}{Assumption}[section]

\newtheorem{rem}{Remark}[section]

\newcommand{\ltwonorm}[1]{\lVert#1\rVert_2}
\newcommand{\fnorm}[1]{\lVert#1\rVert_F}

\def \iht {\mathrm{iht}}
\def \pen {\mathrm{pen}}
\def \tpr {\mathrm{TPR}}
\def \fdr {\mathrm{FDR}}
\def \bss {\mathrm{best}}

\title{Best subset selection is robust against design dependence \footnote{The authors gratefully acknowledges ONR grant N00014-19-1-2120, NSF grant DMS-1662139, NSF grant DMS-2015366 and NIH grant R01-GM072611-16.}}

\author{Yongyi Guo$^{\ast, \diamond}$ \\ \small yongyig@princeton.edu \and Ziwei Zhu$^{\dagger, \diamond}$ \\ \small ziweiz@umich.edu \and Jianqing Fan$^\ast$ \\ \small jqfan@princeton.edu\\
\normalsize
$^\ast$Department of Operations Research and Financial Engineering, Princeton University \\ \normalsize
$^\dagger$Department of Statistics, University of Michigan, Ann Arbor \\ \normalsize
$^\diamond$ The authors contributed equally to this work. }

\begin{document}

\def\spacingset#1{\renewcommand{\baselinestretch}%
{#1}\small\normalsize} \spacingset{1}

%%%%%%%%%%%%%%%%%%%%%%%%%%%%%%%%%%%%%%%%%%%%%%%%%%%%%%%%%%%%%%%%%%%%%%%%%%%%%%

\maketitle

\begin{abstract}
Best subset selection (BSS) is widely known as the holy grail for high-dimensional variable selection. Nevertheless, the notorious NP-hardness of BSS substantially restricts its practical application and also discourages its theoretical development to some extent, particularly in the current era of big data. % Despite its intensive studies, the fundamental question of when BSS is truly the ``best'', namely yielding the oracle estimator that uses the true subset of variables, remains only partially answered. 
In this paper, we investigate the variable selection properties of BSS when its target sparsity is greater than or equal to the true sparsity. Our main message is that BSS is robust against design dependence in terms of achieving model consistency and sure screening, and more importantly, that such robustness can be propagated to the near best subsets that are computationally tangible. Specifically, we introduce an identifiability margin condition that is free of restricted eigenvalues 
% which is natural for finite-sample model identifiability, 
and show that it is sufficient and nearly necessary for BSS to exactly recover the true model. A relaxed version of this condition is also sufficient for BSS to achieve the sure screening property. Moreover, taking optimization error into account, we find that all the established statistical properties for the exact best subset carry over to any near best subset whose residual sum of squares is close enough to that of the best one. In particular, a two-stage fully corrective iterative hard thresholding (IHT) algorithm can provably find a sparse sure screening subset within logarithmic steps; another round of exact BSS within this set can recover the true model. The simulation studies and real data examples show that IHT yields lower false discovery rates and higher true positive rates than the competing approaches including LASSO, SCAD and Sure Independence Screening (SIS), especially under highly correlated design.
\end{abstract}


%\newpage
%\spacingset{1.5} % DON'T change the spacing!
\section{Introduction}
\label{sec:intro}

Variable selection in high-dimensional sparse regression has been one of the most central topics in statistics for decades. Consider $n$ independent and identically distributed (i.i.d.) observations $\{\bx_i, y_i\}_{i=1}^n$ from a linear model:
\beq
\label{eq:lm}
y_i=\bx_i^\top\bbeta^*+\epsilon_i, \quad i\in[n],
\eeq
where $\bx_i$ is a $p$-dimensional design vector, $\epsilon_i$ is random noise that is independent of $\bx_i$ and has sub-Gaussian norm $\lVert\epsilon_i\rVert_{\psi_2}$ bounded by $\sigma$, $\bbeta^*\in\R^p$ and $\norm{\bbeta^*}_0 = s<n$. The major goal of high-dimensional variable selection is to learn the active set of the true regression coefficients, namely $\cS^* := \{j: \beta^*_j \neq 0\}$, when $p$ enormously exceeds $n$.

One well-established principle for high-dimensional variable selection is to penalize empirical risk by model complexity, thereby encouraging sparse solutions. Specifically,  consider
\begin{equation}\label{eq-Mest}
\hat\bbeta^{\mathrm{pen}} :=\argmin_{\bbeta \in \R^p} \cL(\bbeta) + \rho_{\lambda}(\bbeta),
\end{equation}
%\frac{1}{2n}\sum_{i=1}^n (\bx_i^\top\bbeta-y_i)^2+\rho_\lambda(\bbeta) =:
where $\cL(\bbeta)$ is a loss function, and where $\rho_{\lambda}(\bbeta)$ is a model regularizer.
%Note that the regularizer here is decomposible with respect to the components of $\bbeta$, which is a useful property for its statistical analysis \citep[][]{NRW12}.
Classical approaches such as AIC \citep{Aka74, Aka98}, BIC \citep{Sch78} and Mallow's $C_p$ \citep{Mal73} use the model size, i.e., the $L_0$-norm of the regression coefficients, to penalize negative log-likelihood. Though rendering nice sampling properties \citep{BBM99,ZZh12}, such $L_0$-regularized methods are notorious for its computational infeasibility; in general the program has been shown to be NP-hard \citep{FKT15}. The past three decades or so have witnessed massive endeavors on pursuing alternative penalty functions that yield both strong statistical guarantee and computational expediency in the high-dimensional regime. Such efforts have given rise to a myriad of pivotal and powerful methods for variable selection, such as LASSO \citep{Rob96, CDS98, ZY06}, SCAD \citep{FLi01, FPe04, LWa15, LWa17, FLS18}, adaptive LASSO \citep{Zou06}, elastic net \citep{ZHa05}, MCP \citep{Zha10}, among others. We also refer the readers to \cite{BVa11}, \cite{Wai19} and \cite{FLZ20} for comprehensive introduction to recent development in high-dimensional variable selection.

Theoretically, there has been intensive study on when these penalized methods enjoy model consistency, i.e.,  recovering the true model with probability converging to one as $n, p \rightarrow \infty$. Write $\by = (y_1, \ldots, y_n) ^ {\top}$ and $\bX = (\bx_1, \ldots, \bx_n) ^ {\top}$. \cite{ZY06} established the sufficient and nearly necessary conditions for model consistency of the LASSO estimator $\hat\bbeta^{\lasso}$. One of the most crucial conditions involved is the well-known irrepresentable condition, which says that there exists a constant $\eta > 0$ such that
\begin{equation}
\label{eq:irrepresentable}
\left\lVert \hat\bSigma_{(\cS^*)^c\cS^*} (\hat\bSigma_{\cS^*\cS^*})^{-1}\sign(\bbeta^*_{\cS^*})\right\rVert_\infty\leq 1-\eta,
\end{equation}
where $\hat\bSigma_{\cS^*\cS^*}$ is the sample covariance of $\bX_{\cS^*}$ and $\hat\bSigma_{(\cS^*)^c\cS^*}$ is the sample cross covariance between $\bX_{(\cS^*)^c}$ and $\bX_{\cS^*}$. Informally speaking, if we regress any spurious covariate on the true covariates, \eqref{eq:irrepresentable} requires the $\ell_1$-norm of the resulting regression coefficient vector to be bounded by $1 - \eta$, which is generally believed being restrictive in practice:  the bigger the true model, the harder the condition to satisfy.

Nonconvex regularization comes as a remedy for this. It corrects the bias induced by $\ell_1$-regularization, thereby being able to achieve selection consistency without the irrepresentable condition \citep{fan2011nonconcave}. Let $\mu_*:= \min_{j \in \cS^*}|\beta^*_j|$. \cite{Zha10} shows that when $\mu_* \gtrsim \sqrt{\log p / n}$, MCP enjoys selection consistency under a sparse Riesz condition on $\bX$, i.e.,
\[
	0< c_{\ast} \le \min_{|\cA| \le m} \lambda_{\min}(\bSigma_{\cA\cA}) \le \max_{|\cA| \le m} \lambda_{\max}(\bSigma_{\cA\cA}) \le c^{\ast} < \infty,
\]
where $\bSigma_{\cA\cA}$ is the population covariance of $\bX_{\cA}$, and where $m \gtrsim s$. \cite{FLS18} propose an iterative local adaptive majorize-minimization (I-LAMM) algorithm for empirical risk minimization with folded concave penalty. Under a general likelihood framework, they show that only a local Riesz condition suffices to ensure model consistency. Specifically, for any sparsity $m \in [p]$ and neighborhood radius $r$, define the maximum and minimum localized sparse eigenvalues (LSE) of $\nabla^2 \cL$ around $\bbeta^*$ as follows:
\begin{equation}
\label{eq:lse}
\begin{aligned}
\rho_+(m, r):=\sup_{\bu, \bbeta}\left\{\bu_J^\top\nabla^2\cL(\bbeta)\bu_J:\lVert \bu_J\rVert_2^2=1,\lvert J\rvert\leq m, \lVert\bbeta-\bbeta^*\rVert_2\leq r\right\},\\
\rho_-(m, r):=\inf_{\bu, \bbeta}\left\{\bu_J^\top\nabla^2\cL(\bbeta)\bu_J:\lVert \bu_J\rVert_2^2=1,\lvert J\rvert\leq m, \lVert\bbeta-\bbeta^*\rVert_2\leq r\right\}.
\end{aligned}
\end{equation}
I-LAMM is proved to enjoy model consistency if $\rho_+$ and $\rho_-$ are bounded from above and below respectively with $m \asymp s$, $r \asymp \sqrt{s\log p / n}$ and $\mu_{\ast} \allowbreak \gtrsim \sqrt{\log p / n}$. Nevertheless, sufficient and nearly necessary conditions for nonconvex penalized methods to achieve model consistency have yet been found.

%For notational convenience, let $\by := (y_1, \ldots, y_n)^\top$, $\bX := (\bx_1, \ldots, \bx_n)^\top$ and $\cS^*:=\supp(\bbeta^*)$.

Recent advancement in algorithms and hardware has sparked a revival of interest in the best subset selection (BSS) despite its computational hardness. \cite{BKM16} propose and study a Mixed Integer Optimization (MIO) approach for solving the classical BSS problem, i.e.,
\begin{equation}
\label{eq:best_subset}
\hat \bbeta^{\bss} (\hat s):= \argmin_{\bbeta \in \R^p, \|\bbeta\|_0 \le \hat s} \cL(\bbeta),
\end{equation}
where $\hat s$ is an estimator of the sparsity. In the sequel, for conciseness we drop $\hat s$ when we write $\hat \bbeta^{\bss}(\hat s)$. They show that the MIO algorithm can find a near-optimal solution of \eqref{eq:best_subset} within minutes when $n$ is in the $100$s and $p$ is in the $1000$s. Their simulations also suggest that when a spurious predictor is highly correlated with a true predictor in the high-dimensional setup, LASSO tends to select a dense model and thus yields much worse prediction performance than the MIO (see Fig. 8 therein). A recent follow-up work \cite{bertsimas2020sparse} proposed a new cutting plane method that solves the BSS problem with Ridge penalty with $n, p$ in the $100,000$s. \cite{HTT17} expand the simulation experiments of \cite{BKM16} and show that in terms of the prediction risk, BSS performs better than LASSO when the signal-to-noise ratio (SNR) is high, while performing worse than LASSO when the SNR is low. These works motivate us to systematically investigate the variable selection properties of BSS and compare them with those of LASSO and SCAD. Unlike Lasso, there is no sufficient and nearly necessary condition for BSS to achieve model selection consistency.
\cite{SPZ12} and \cite{SPZ13} are among the earlier papers on the variable selection properties of BSS. They establish the optimality of BSS in terms of variable selection, in the sense that it achieves model consistency under a ``minimal separation condition''.
% that is proved to be necessary for model consistency. Specifically, \cite{SPZ13} define the following degree of separation to characterize the difficulty of high-dimensional variable selection:
%\[
%C_{\min}(\bbeta ^ *, \bX) := \min_{\substack{|\cS| \le s, \bbeta \in \RR^p, \\ \supp(\bbeta) = \cS}} \frac{1}{n\max(|\cS^*\backslash \cS|, 1)} \ltwonorm{\bX\bbeta ^ * - \bX\bbeta} ^ 2,
%\]
%where $\bX = (\bx_1, \bx_2, \ldots, \bx_n) ^ \top$. 
%They show that the selection consistency requires that $C_{\min}(\bbeta ^ *, \bX) \gtrsim \sigma ^ 2 \log p /n$, where $\sigma  := \sqrt{\var(\epsilon_1)}$, and that $\hat\bbeta^{\bss}(s)$ and 
They show further that its computational surrogate based on truncated $\ell_1$ penalty (TLP) consistently  recovers $\cS^*$.
%when $C_{\min}(\bbeta ^ *, \bX) \gtrsim \sigma ^ 2 \log p /n$.

In this paper, we focus on the model selection properties of BSS and a two-stage fully corrective iterative hard thresholding (IHT) algorithm that provably solves the BSS problem with relaxed sparsity constraint \citep{JTK14}. More specifically, this IHT algorithm can find a solution $\hat\bbeta^{\iht}$ with sparisity slightly larger than $\hat s$, such that $\cL(\hat\bbeta^{\iht})$ is below $\cL(\hat\bbeta^{\bss}(\hat s))$, which is the minimum of the objective function in the best $\hat s$-subset selection problem. Based on this optimization result, we establish the model selection properties of $\hat\bbeta^{\iht}$. In the analysis, we need to take into account both statistical and optimization error in an non-asymptotic manner, which distinguishes our work from \cite{SPZ12} and \cite{SPZ13}. Given an estimator $\widehat\bbeta$, define its true positive rate (TPR) as
\[
	\tpr(\widehat\bbeta) := \frac{|\supp(\widehat\bbeta) \cap \cS^*|}{|\cS^*|},
\]
and define its false discovery rate (FDR) as
\[
	\fdr(\widehat \bbeta) := \frac{|\supp(\widehat\bbeta) \cap(\cS^*)^{c}|}{\max(|\supp(\widehat\bbeta)|, 1)}.
\]
Our major contributions are threefold:
\begin{enumerate}
		\item We identify a crucial quantity, called
        \emph{the identifiability margin}, that determines whether $\widehat\bbeta^{\bss}$ or its \emph{approximate solution} achieves exact model recovery. This quantity is independent of the restricted eigenvalues of the design, thereby accommodating highly dependent design. The sufficient and necessary condition for BSS to achieve model consistency boils down to a lower bound of the identifiability margin, which is weaker and more natural than the $\beta$-min condition in \cite{ZZh12}, and which is weaker than the LSE condition of \citet{FLS18}. See Theorems \ref{thm-selection-consistency} and \ref{thm:lower_bound}.
        
        %Specifically, for any $\cS\subset\{1,\cdots, p\}$ with $|\cS|=s$,
		%\beq
		%	\label{eq:schur_complement}
		%	\hat\bD(\cS):=\hat\bSigma_{\cS^*\setminus\cS, \cS^*\setminus\cS}-\hat\bSigma_{\cS^*\setminus\cS, \cS}\hat\bSigma_{\cS \cS}^{-1}\hat\bSigma_{\cS, \cS^*\setminus\cS}.
		%\eeq
		%$\hat \bD(\cS)$ can be regarded as the empricial conditional covariance of $\bx_{\cS^* \backslash \cS}$ given $\bx_{\cS}$. Define $\hat\lambda_m := \min_{|\cS| = s, \cS \neq \cS ^ *} \lambda_{\min}(\hat\bD(\cS))$. Theorems \ref{thm-selection-consistency} and \ref{thm:lower_bound} show that the lower $\hat\lambda_{m}$, the harder for BSS to identify the true model $\cS ^ *$ (see Remark \ref{rem:lambda_m} for the details).
		
		\item We explicitly characterize $\tpr(\hat\bbeta^{\bss})$ when the sparsity is overestimated (see Theorem \ref{thm:power} for the details). The identifiability margin also plays a critical role in guaranteeing the sure screening property of any reasonable approximate solution to the BSS.  In particular, we show that the more we overestimate $s$, the stronger signal is required to guarantee sure screening of $\hat\bbeta^{\bss}$, i.e., $\tpr(\hat\bbeta^{\bss}) = 1$.
		
		\item We study a two-stage fully corrective IHT algorithm and provide a TPR guarantee of its solution $\hat\bbeta^{\iht}$. If the true sparsity $s$ is known, a further application of BSS on the support of $\hat\bbeta^{\iht}$ can yield exactly the true model. Our simulations demonstrate that $\hat\bbeta^{\iht}$ exhibits remarkably higher TPR than LASSO and SCAD at the same level of FDR, especially in presence of strong correlation of design.
	\end{enumerate}

	The rest of the paper is organized as follows. Section \ref{sec:bss} analyzes the model selection properties of BSS when the sparsity is either known or overestimated. Section \ref{sec:iht} introduces the IHT algorithm and establishes the TPR guarantee of its solution. Section \ref{sec:sim} compares the TPR-FDR curve of IHT with those of LASSO, SCAD and SIS under different signal-to-noise ratios and correlation structure of the design. Finally, Section \ref{sec:real} analyzes two real datasets on diabetes and macroeconomics respectively to illustrate the power of the IHT algorithm in model selection.

\section{Model selection properties of BSS}\label{sec:bss}

\subsection{Model consistency of BSS with known sparsity}

%The finite-sample model identifiability margin can be measured as follows.  
Let $\by = (y_1, \ldots, y_n) ^ {\top}$, $\bX = (\bx_1, \ldots, \bx_n) ^ {\top}$ and $\bepsilon = (\epsilon_1, \ldots, \epsilon_n) ^ {\top}$. Let $\bX_{\cS}$ denote the matrix comprised of only the columns of $\bX$ with indices in $\cS$, and let $\bP_{\bX_{\cS}} := \bX_{\cS}(\bX^\top_{\cS} \bX_{\cS})^{-1} \bX_{\cS}^\top$ denote the projection matrix corresponding to the column space of $\bX_{\cS}$. Given any candidate model $\cS \subset [p]$ other than $\cS ^ *$, we can rewrite model \eqref{eq:lm} in the matrix form: 
$$
\by = \bX_{\cS^*} \bbeta_{\cS^*}^* + \bepsilon
= \bP_{\bX_{\cS}} \bX_{\cS^*} \bbeta_{\cS^*}^* + (\bI_n - \bP_{\bX_{\cS}})\bX_{\cS^*} \bbeta_{\cS^*}^* + \bepsilon. 
$$
The term $(\bI_n - \bP_{\bX_{\cS}})\bX_{\cS^*} \bbeta_{\cS^*}^*$ is the part of the signal that cannot be linearly explained by $\bX_{\cS}$. We can thus measure the discrimination margin between models $\cS$ and $\cS^*$ through
$$
\ltwonorm{(\bI_n - \bP_{\bX_{\cS}})\bX_{\cS^*} \bbeta_{\cS^*}^*} ^ 2 = \bbeta_{\cS ^ *}^{\top} \bX_{\cS ^ *} ^ {\top} (\bI_n - \bP_{\bX_{\cS}})\bX_{\cS ^ *}\bbeta ^*_{\cS ^ *} = \bbeta_{\cS_0}^{*\top} \bX_{\cS_0} ^ {\top} (\bI_n - \bP_{\bX_{\cS}})\bX_{\cS_0}\bbeta ^*_{\cS_0},
$$
where $\cS_0 := \cS ^ * \setminus \cS$. Let $\hat\bSigma := n^{-1}\bX^\top\bX$ be the sample covariance matrix, and for any two sets $\cS_1, \cS_2\subset\{1,2,\cdots, p\}$, let $\hat\bSigma_{\cS_1,\cS_2}$ be the submatrix of $\hat\bSigma$ containing the intersection of the rows indexed in $\cS_1$ and columns indexed in $\cS_2$. Note that if we define 
\beq
\label{eq:schur_complement}
\hat\bD(\cS):=\hat\bSigma_{\cS_0, \cS_0}-\hat\bSigma_{\cS_0, \cS}\hat\bSigma_{\cS \cS}^{-1}\hat\bSigma_{\cS, \cS_0}, 
\eeq
which is the covariance of the residuals of $\bX_{\cS_0}$ after being linearly regressed on $\bX_{\cS}$, then we have that $\ltwonorm{(\bI_n - \bP_{\bX_{\cS}})\bX_{\cS^*} \bbeta_{\cS^*}^* } ^ 2 = n\bbeta^{*\top}_{\cS_0} \hat \bD(\cS) \bbeta_{\cS_0} ^*$. Intuitively, $\cS ^ *$ is identifiable only when $\bbeta^{*\top}_{\cS_0} \hat \bD(\cS) \bbeta_{\cS_0} ^*$ is distinctively large for all $\cS \neq \cS ^ *$. This leads to the definition of the \emph{identifiability margin} in Theorem \ref{thm-selection-consistency}, the crucial quantity that determines whether BSS can achieve model consistency. 

%in comparison with $|\cS_0^*| \sigma^2$ (recalling $F$-test with infinite degree of freedom on the denominator), where $\cS_0  = \cS^* \setminus \cS$ 
%This should be compared for all $\cS$ under consideration and leads to the identifiability margin defined in the theorem below.  reflecting the part of $\bX_{\cS_0}$ that can not be explained by $\bX_{\cS}$.  
%Hence, $\bbeta_{\cS_0} ^{*\top} \hat \bD(\cS) \bbeta_{\cS_0}^*$ measures the part $\bX_{\cS_0}^\top \bbeta_{\cS_0}^*$ that can not be linearly explained by $\bX_{\cS}$.

Now we are in position to present our first theoretical result.  For any set $\cS\subseteq [p]$, define the sum of squared residuals $R_{\cS}$ of $\by$ on $\bX_{\cS}$ as
\[
R_{\cS} := \by^\top\bigl\{\bI - \bX_{\cS}(\bX^\top_{\cS} \bX_{\cS})^{-1} \bX_{\cS}^\top\bigr\}\by = \by^\top(\bI-\bP_{\bX_{\cS}})\by. 
\] 
In addition, for any sparsity estimate $\hat s$, define 
\[
	\cA(\hat s):=\{\cS\subset [p]: |\cS|=\hat s, \cS\neq \cS^*\}, 
\]
which represents the set of all models of size $\hat s$ except the true one. 
The following theorem gives a sufficient condition for BSS to recover exactly the true model for fixed designs. 
%Moreover, for any $\cS\subset\{1,\cdots, p\}$ such that $|\cS|=s$, define
%$$
%\hat\bD(\cS):=\hat\bSigma_{\cS^*\setminus\cS, \cS^*\setminus\cS}-\hat\bSigma_{\cS^*\setminus\cS, \cS}\hat\bSigma_{\cS \cS}^{-1}\hat\bSigma_{\cS, \cS^*\setminus\cS}.
%$$

\begin{thm}\label{thm-selection-consistency}
For any $p \geq 3$ and sparisty estimate $\hat s$, define the identifiability margin 
\begin{equation}
	\label{eq:im1}
	\tau_*(\hat s) := \min_{\cS \in \cA(\hat s)} \frac{\bbeta_{\cS^* \setminus \cS} ^{*\top} \hat \bD(\cS) \bbeta_{\cS^* \setminus \cS}^*}{|\cS \setminus \cS ^ *|}. 
\end{equation}
Then there exists a universal constant $C > 1$, such that for any $\xi > C$ and $0\le  \eta < 1$, whenever
\beq
	\label{eq:thm2.1_tau_condition}
	\tau_*(s) \ge \biggl(\frac{4\xi}{1 - \eta}\biggr) ^ {2}\frac{\sigma ^ 2\log p}{n},
\eeq
%\beq
%	\label{eq:min_marginal_signal}
%    \mu_* \geq \frac{4\xi \sigma}{1 - \eta}\biggl(\frac{\log p}{n\hat \lambda_m}\biggr)^{1 / 2},
%\eeq
we have with probability at least $1 - 8sp^{- (C^{-1}\xi - 1)}$ that
\beq
	\Bigl\{\widehat \cS: %\widehat \cS \subset [p],
 |\widehat \cS| =s, R_{\widehat \cS} \le \min_{\cS \subset [p], |\cS| =s} R_{\cS} + n\eta\tau_*(s) \Bigr\} = \{\cS^*\},
\eeq
which, in particular, implies that $\cS^* = \argmin_{\cS\subset[p], |\cS|=s} R_{\cS} $.
\end{thm}

Theorem~\ref{rem:lambda_m} asserts that if the identifiability margin of the true model satisfies \eqref{eq:thm2.1_tau_condition}, any estimator $\hat \cS$ with optimization error within $n\eta\tau_*(s)$ selects the true model.

%\begin{rem}
%	Consider the population counterpart $\bD(\cS)$ of $\widehat \bD(\cS)$:
%	\[
%	\bD(\cS) := \bSigma_{\cS^*\setminus\cS, \cS^*\setminus\cS} - \bSigma_{\cS^*\setminus\cS, \cS}\bSigma_{\cS \cS}^{-1}\bSigma_{\cS, \cS^*\setminus\cS}.
%	\]
%	Note that when $\bx_1$ follows a multivariate normal distribution, $\bD_\cS$ is the conditional covariance matrix of $[\bx_1]_{\cS^* \setminus \cS}$ given $[\bx_1]_{\cS}$, i.e.,
%	\beq
%	\label{eq:conditional_cov}
%	\bD(\cS) = \Cov\bigl([\bx_1]_{\cS^* \setminus \cS}~|~[\bx_1]_{\cS}\bigr).
%	\eeq
	%		Therefore, $\widehat \lambda_m = 0$ corresponds to the case where one can find a $\cS \neq \cS^*$ such that $\bx_{\cS \cup \cS^*}$ is degenerate, which we do not expect to occur in practice.
%\end{rem}

\begin{rem}
	\label{rem:lambda_m}
	Condition \eqref{eq:thm2.1_tau_condition} is more natural and weaker than the $\beta$-min condition in \cite{ZZh12}. Let $\mu_* := \min_{j \in [p]} |\beta^*_j|$ and $\hat \lambda_m:= \min_{\cS \in \cA(s)} \lambda_{\min}(\hat \bD(\cS))$. Then, $\bbeta_{\cS^* \setminus \cS} ^{*\top} \allowbreak \hat \bD(\cS) \bbeta_{\cS^* \setminus \cS}^*
	%\ge \lambda_{\min}(\hat \bD(\cS)) \ltwonorm{\bbeta^*_{\cS \backslash \cS ^ *}} ^ 2
	\ge \hat \lambda_m |\cS^* \backslash \cS|\mu_* ^ 2$ and hence $\tau_*(s) \geq \lambda_{\min}(\hat \bD(\cS)) \mu_* ^ 2$. Therefore, a sufficient condition for \eqref{eq:thm2.1_tau_condition} is that
	\beq
		\label{eq:beta_min}
		\mu_* \ge \frac{4\xi\sigma}{1 - \eta} \biggl(\frac{\log p}{n\hat\lambda_m}\biggr)  ^ {1 / 2}.
	\eeq
	 \citet{ZZh12} showed that the $\ell_0$-regularized least squares estimator is able to achieve model consistency when $\mu_* \gtrsim \sigma\sqrt{\log p / (n\kappa_-)}$, where $\kappa_-:= \min_{\cA: |\cA| \le s, \cA \subset [p]} \allowbreak\lambda_{\min}(\bSigma_{\cA\cA})$. Therefore, the major difference between this condition and \eqref{eq:beta_min} lies in the difference between $\kappa_-$ and $\hat \lambda_m$.  Note that $\hat \lambda_{m}$ is insensitive to the collinearity between spurious covariates themselves; rather, it reflects how spurious variables can approximate the true model, which implies much less restriction than that induced by $\kappa_-$. To further illustrate this point, consider $100$ standard Gaussian covariates $\{X_j\}_{j \in [100]}$. Suppose that $ \cS ^ * = \{1, 2\}$, i.e., the true model has only two covariates $X_1$ and $X_2$. Let $\cov(X_j, X_k) = r > 0$ for any $j, k \ge 3$ such that $j \neq k$, and let $\cov(X_1, X_j) = \cov(X_2, X_j) = 0$ for any $j \ge 3$. In words, all the spurious covariates are correlated with each other but independent of the true covariates. One can then verify that $\hat\lambda_m = 1$, but $\hat\kappa_- = 1 - r$.
%	 \jf{be careful}. 
	 This implies that as $r \to $ 1, the $\beta$-min condition of \citet{ZZh12} requires higher and higher signal strength, whereas our condition \eqref{eq:beta_min} does not at all! Therefore, \eqref{eq:beta_min} shows the robustness of BSS against dependence between spurious variables.
	
\end{rem}

\begin{rem}
		Condition \eqref{eq:beta_min} is also weaker than the LSE condition of \citet{FLS18}. \eqref{eq:beta_min} allows $\widehat \lambda_m$ to decrease to $0$ as $n$ and $p$ grow; this scenario, however, implies that $\rho_-(2s, r)$ in \eqref{eq:lse} converges to $0$ uniformly over $r \in \mathbb{R}$ and thus contradicts the LSE condition in \cite{FLS18}. To see this, write $\cS_0  = \cS^* \setminus \cS$.  Since  $\widehat \lambda_m \rightarrow 0$, for any $\epsilon > 0$,  there exist $\cS \subset [p]$ with $|\cS| \le s$ and $\bv \in \mathbb{R}^{|\cS_0|}$ such that
		 \[
		 \bv^\top \widehat \bSigma_{\cS_0, \cS_0} \bv - \bv^\top \widehat\bSigma_{\cS_0, \cS}\widehat \bSigma^{-1}_{\cS, \cS}\widehat\bSigma_{\cS, \cS_0} \bv \le \epsilon \|\bv\|_2^2.
		 \]
		 Construct $\widetilde \bv = (\bv^\top, - \bv^\top \widehat \bSigma_{\cS_0, \cS} \widehat \bSigma_{\cS, \cS}^{-1})^\top \in \RR^{|\cS_0 \cup \cS|}$. Then the inequality above yields that for any $\bbeta \in \RR^p$,
		 \[
		 \widetilde \bv^\top \nabla^2 \cL(\bbeta)\widetilde \bv = \widetilde \bv^\top \widehat \bSigma_{\cS_0 \cup \cS, \cS_0 \cup \cS} \widetilde \bv = \bv^\top (\widehat \bSigma_{\cS_0, \cS_0} - \widehat \bSigma_{\cS_0, \cS} \widehat \bSigma^{-1}_{\cS, \cS} \widehat \bSigma_{\cS, \cS_0}) \bv \le \epsilon \ltwonorm{\widetilde \bv}^2.
		 \]
		 Therefore, $\rho(2s, r) \le \epsilon$ for all $r > 0$ and our claim follows by arbitrariness of $\epsilon$. This example illustrates the capability of the identifiability margin to leverage the signal strength to overcome collinearity of the design.
	\end{rem}

\begin{rem}
	Finally, we discuss the relationship between the condition \eqref{eq:thm2.1_tau_condition} and the irrepresentable condition in \cite{ZY06}. Though BSS outperfoms LASSO in terms of model selection in general as illustrated in our numerical study, one cannot deduce \eqref{eq:thm2.1_tau_condition} from the irrepresentable condition. In other words, there are some \emph{corner cases} where LASSO can recover the true model, while BSS cannot. For example, suppose there are three four-dimensional observations: $\bX = [(1 + \eta ^ 2)^{-1 / 2}(\be_1 + \eta \be_3), (1 + \eta ^ 2 )^{-1 / 2}(\be_1 - \eta\be_3), 2^{-1 / 2}(\be_1 + \be_2), \be_2] \in \RR^{3 \times 4}$, where $\eta < 1$ and $\be_j$ is the $j$th canonical basis vector. The true model is that $Y = (1 + \eta ^ 2)^{1 / 2}(X_1 + X_2) / 2$, which implies that $\cS^* = \{1, 2\}$. Some algebra yields that
	\[
	\left\lVert \hat\bSigma_{(\cS^*)^c\cS^*} (\hat\bSigma_{\cS^*\cS^*})^{-1}\sign(\bbeta^*_{\cS^*})\right\rVert_\infty = \biggl(\frac{1 + \eta ^ 2}{2}\biggr)^{1 / 2} < 1.
	\]
	Therefore, the irrepresentable condition is satisfied, and LASSO is able to recover the true model. In contrast, BSS cannot recover $\cS ^ *$, because $\bX_1 + \bX_2$ is parallel to $\bX_3 - 2^{-1 / 2}\bX_4$, and thus $\lambda_{\min}(\bD(\{3, 4\})) = 0$. The root reason for BSS's failure to capture the true model is that the $\ell_0$ constraint does not have any preference between the models $\{1, 2\}$ and $\{3, 4\}$, while LASSO prefers $\{1, 2\}$ because the resulting regression coefficients have smaller $\ell_1$-norm. Of course, if the true model is $\{3, 4\}$, LASSO will choose the wrong model.
	%	\[
	%		\bX = \left(
	%			\begin{array}{ccc}
	%				
	%			\end{array}
	%			\right).
	%	\]
\end{rem}

Theorem \ref{thm-selection-consistency} shows that the identifiability margin determines the model consistency of BSS. A natural question then arises: is the requirement \eqref{eq:thm2.1_tau_condition} on the identifiability margin necessary for such model consistency? The following theorem shows that it is almost necessary by giving a necessary condition that takes a similar form as \eqref{eq:thm2.1_tau_condition}. For any $\cB\subset \R^n$ and $\delta > 0$, let $M(\delta, \cB)$ denote the $\delta$-packing number of $\cB$ under Euclidean distance. We first introduce a technical assumption that excludes extremely correlated setups.

\begin{ass}\label{ass:lower_bound}
	There exist $j_0 \in \cS^*$, a universal constant $0 < \delta_0 < 1$ and $c_{\delta_0} > 0$ such that if we let $\cS_0^* := \cS^* \setminus j_0$, $\widetilde \bu_j:=(\bI-\bP_{\bX_{\cS_0^*}})\bX_j$ and $\overline \bu_j := \widetilde \bu_j / \ltwonorm{\widetilde \bu_j}$ for $j \in [p] \setminus \cS^*$,  then $\log M(\delta_0, \{\overline \bu_j\}_{j \in [p] \setminus \cS^*})\geq c_{\delta_0} \log p$.
\end{ass}

Basically, Assumption \ref{ass:lower_bound} says that there are $\Omega(p ^ {c_{\delta_0}})$ spurious variables that are not too correlated with each other. Violating this assumption means that all the spurious variables are highly correlated with each other, in which case the lower bound  \eqref{eq:thm2.1_tau_condition} for the identifiability margin is not necessary to find the true model. Now we are in position to introduce the necessary condition.

\begin{thm}
	\label{thm:lower_bound}
	Suppose Assumption \ref{ass:lower_bound} holds. Furthermore, assume $\{\epsilon_i\}_{i \in [n]}$ are i.i.d. Gaussian noise with variance $\sigma ^ 2$. Consider all the models of size $s$ that are formed by replacing $j_0$ in $\cS ^ *$ with a spurious variable, namely, $\cC_{j_0} := \{\cS\subset [p]: |\cS|=s, \cS^*\setminus \cS =\{j_0\}\} \subset \cA(s)$. Define the maximum leave-one-out identifiability margin as
	\beq
		\label{eq:tau_max}
		\tau^* := \max_{\cS \in \cC_{j_0}} \frac{\bbeta_{\cS^* \setminus \cS} ^{*\top} \hat \bD(\cS) \bbeta_{\cS^* \setminus \cS}^*}{|\cS \setminus \cS ^ *|}  = \max_{\cS \in \cC_{j_0}} \widehat D(\cS){\beta^{*}_{j_0}}^2.
	\eeq
	Then there exist $c, C_1 > 0$, depending on $\delta_0$ in Assumption \ref{ass:lower_bound}, such that whenever  $\tau^*<c\sigma^2\log p / n$, with probability at least $1-C_1(\log p) ^ {-1} -2p^{-1}$, $\cS^* \notin \argmin_{\cS\subset[p], |\cS|=s} R_{\cS}$.
\end{thm}

Theorem \ref{thm:lower_bound} shows that under Assumption \ref{ass:lower_bound}, if the maximum leave-one-out identifiability margin violates the lower bound in \eqref{eq:thm2.1_tau_condition}, then with high probability we fail to recover the true model.

\subsection{Sure screening of BSS with overestimated sparsity}

In this section, we study the model selection property of the best subset selection when the model sparsity is overestimated, i.e., $\hat s > s$. In this scenario, it is impossible for BSS to achieve exact recovery of the true model, but a desirable property to have is that all the true variables are selected, i.e. $\tpr(\cS) = 1$. We call this the sure screening property. Sure screening \citep{FLv08} eliminates spurious variables and allows us to recover the true model from a much smaller pool of predictors. The following theorem characterizes when BSS achieves the sure screening property.

\begin{thm}
	\label{thm:power}		
	Suppose that $\hat s\geq s$, $p \ge 3$, and that the design is fixed. For any $\delta \in (0, 1]$, define more generally the identifiability margin as
	\beq
		\label{eq:tau_delta}
		\tau_*(\hat s, \delta) := \min_{\cS \in \cA(\hat s), |\cS ^ * \setminus \cS| \ge \delta s} \frac{{\bbeta^*_{\cS^* \backslash \cS}}^\top \hat\bD({\cS})\bbeta^*_{\cS^* \backslash \cS}}{|\cS \backslash \cS ^ *|}. 
	\eeq
%	Following the same notation as in Theorem \ref{thm-selection-consistency}, define
%			\[
%			\tau_*(\hat s) := \min_{\cS \in \widehat\cA} \frac{{\bbeta^*_{\cS^* \backslash \cS}}^\top \hat\bD({\cS})\bbeta^*_{\cS^* \backslash \cS}}{|\cS^* \backslash \cS|}.
%			\]
%			Moreover, define
%			\[
%			\tilde\tau_*(\hat s) := \min_{\cS \in \widehat\cA} \frac{{\bbeta^*_{\cS^* \backslash \cS}}^\top \hat\bD({\cS})\bbeta^*_{\cS^* \backslash \cS}}{|\cS \backslash \cS^*|}.
%			\]
		Then there exists a universal constant $C > 1$, such that for any $\xi > C$ and $0 \le  \eta < 1$, whenever
		\beq
			\label{eq:min_tau_condition_1}
			\tau_*(\hat s, \delta)\geq \biggl(\frac{4\xi }{1 - \eta}\biggr)^2\frac{\sigma^2\log p}{n},
		\eeq
	 	we have that
	 	\[
	 		\PP\biggl(\tpr(\cS)\geq 1-\delta,~\forall \cS~\text{s.t.}~|\cS|=\widehat s \text{ and } R_{\cS} \le R_{\cS^*} + n\eta\tau_*(\hat s, \delta)\biggr) \ge 1 - 8sp^{- (C^{-1}\xi - 1)}.
	 	\]
	 	In particular, when \eqref{eq:min_tau_condition_1} holds for some $\delta < s ^ {-1}$, we have that
	 	\beq
	 	    \label{eq:28}
			\PP\biggl(\tpr(\cS) = 1,~\forall \cS~\text{s.t.}~|\cS|=\widehat s \text{ and } R_{\cS} \le R_{\cS^*} + n\eta\tau_*(\hat s, \delta)\biggr) \ge 1 - 8sp^{- (C^{-1}\xi - 1)}.
		\eeq 	
\end{thm}		

\begin{rem}
    \cite{xiong2014better} also studies the sure screening property of BSS when the  true sparsity is overestimated. The assumptions therein essentially require $\tau_*(\widehat s, \delta)$ to be of order at least $n ^ {- 1 / 2}$, which is more restrictive than our lower bound of order $(\log p / n) ^ {-1}$ in \eqref{eq:28}. Besides, \cite{xiong2014better} assumes Gaussian noise, while the theorem above accommodates all sub-Gaussian noise.
\end{rem}	 

Theorem \ref{thm:power} can be regarded as a generalization of Theorem \ref{thm-selection-consistency}. We can deduce Theorem \ref{thm-selection-consistency} from Theorem \ref{thm:power} by setting $\hat s =s$ and $\delta = 0$. Besides, note that $\tau_*(\hat s, \delta)$ is a monotonically increasing function with respect to $\delta$. Therefore, a larger $\delta$ implies that the condition \eqref{eq:min_tau_condition_1} is weaker, which corresponds to weaker TPR guarantee. Finally, if we are able to obtain $\widehat \bbeta ^ {\bss}(\hat s)$ exactly, the resulting set of selected variables $\hat \cS$ satisfies that $R_{\hat \cS} \le R_{\cS ^ *}$ and thus enjoys the established TPR guarantee. However, the pursuit of the exact solution is unnecessary: the requirement that $R_{\cS} \le R_{\cS ^ *} + n\eta\tau_*(\hat s, \delta)$ suggests that a good approximated solution to the best $\hat s$-subset selection problem suffices to achieve the TPR guarantee. The next section shows that an IHT algorithm can provide such a qualified approximation.

\section{Iterative hard thresholding} \label{sec:iht}

This section introduces a two-stage iterative hard thresholding (IHT) algorithm that approximately solves the BSS problem and enjoys the TPR guarantee as characterized in Theorem \ref{thm:power}. IHT is a popular family of algorithms for compressed sensing and sparse regression \citep{BDa08, BDa09, JTK14}. The original IHT algorithm \citep{BDa08, BDa09} is essentially $\ell_0$-norm projected gradient descent: in each iteration, the algorithm applies hard thresholding on each coordinate of the gradient to sparsify the gradient and then performs a step of sparse gradient descent. \cite{BDa08} showed that this IHT algorithm is able to converge to local optima of the BSS problem. Nevertheless, this result does not suffice for deriving the variable selection properties of IHT because of lack of statistical guarantee for these local optima. One more recent work \citet[][]{JTK14} investigated a two-stage fully corrective IHT algorithm, which is similar to Compressed Sampling Matching Pursuit (CoSaMP) proposed by \cite{needell2009cosamp}. They showed that when the loss function $\cL(\bbeta)$ in \eqref{eq:best_subset} satisfies the restricted strong convexity (RSC) and restricted strong smoothness (RSS) conditions, this two-stage IHT algorithm is able to achieve a lower objective value than the global minimum of \eqref{eq:best_subset} (i.e., $\cL(\hat\bbeta^{\bss}(\hat s)$) by selecting slightly more than $\hat s$ variables. Despite the sparsity relaxation there, which is inevitable given the NP-hardness of the problem, this result directly targets the global optimum and thus unlocks the potential for studying the model selection properties of the iterates of IHT. 

In this section, we focus on the two-stage fully corrective IHT that is analyzed in \cite{JTK14}. We aim to address the following question: can this IHT algorithm inherit the model selection properties we establish in Section \ref{sec:bss}? Apparently, one cannot expect it to achieve model consistency, because it needs to select a larger model than the true one to ensure the goodness of fit \citep[][Theorem~4]{JTK14}. Therefore, our main interest here is to see whether IHT enjoys the sure screening property as established in Theorem \ref{thm:power}, or more generally, to assess the TPR of the solution of IHT. In the sequel, we first formally introduce the IHT algorithm and then establish its TPR guarantee.

\subsection{Algorithm}
Here we introduce the two-stage fully corrective IHT algorithm in \cite{JTK14}. For any $\bv \in \R^p$ and $r \in \N$, let
\[
\cT_{\mathrm{abs}}(\bv, r) := \bigl\{j: |v_j| \text{ is among the top }r\text{ largest values of }\{|v_k|\}_{k = 1}^p \bigr\}.
\]
The pseudocode of the algorithm is presented in Algorithm \ref{alg-IHT2}. As the name indicates, this IHT algorithm has two stages in each iteration: variable recruiting and elimination. It first recruits the variables that correspond to the largest components of the gradient. When $\cL$ is square loss, the gradient is the covariance between the residuals $\by - \bX\hat\bbeta_t$ and the predictors $\bX$. Hence, this recruiting stage can be interpreted as pulling in the variables with the highest marginal explanation power for the residuals. Then, the algorithm fits an OLS on the resulting expanded model (``fully corrective'' step). Finally, the algorithm eliminates the variables with small coefficients in the OLS (the second stage), so that the sparsity of the model reduces back to $\pi$. In a nutshell, Algorithm \ref{alg-IHT2} alternates between forward and backward selection until the model selection becomes stationary. Once the convergence threshold is hit, Algorithm \ref{alg-IHT2} adjusts the model size to be $\hat s$ by another round of variable recruiting or elimination that is similar to that inside the loop.

For brevity, we refer to Algorithm \ref{alg-IHT2} as simply IHT from now on.

\begin{algorithm}[t]
	\caption{Two-Stage Fully Corrective IHT}	
	\label{alg-IHT2}
	\begin{algorithmic}[1]		
		\STATE \textbf{{Input}}:
		Initial value ${\widehat \bbeta}_{0}=\bzero$, projection size $\pi$, expansion size $l$, sparsity estimate $\hat s$, convergence threshold $\tau > 0$.
		\STATE $t \leftarrow 0$
		% \Do
		\REPEAT
		\STATE $\cG_t \leftarrow \topabs(\nabla\cL({\widehat \bbeta_t}), l)$
		\STATE $\cS_t \leftarrow \supp({\widehat \bbeta}_t) \cup \cG_t$
		\STATE $\widehat \bbeta^\dag_t \leftarrow (\bX_{\cS_t}^\top \bX_{\cS_t})^{-1}\bX_{\cS_t}^\top \by$
		\STATE $\cS^\dag_t \leftarrow \topabs(\widehat\bbeta^\dag_t, \pi)$
		\STATE $\widehat \bbeta_{t + 1} \leftarrow (\bX_{\cS^\dag_t}^\top \bX_{\cS^\dag_t})^{-1}\bX_{\cS^\dag_t}^\top \by$
		\STATE $t \leftarrow t + 1$
		%	\doWhile {$\|\widehat \bbeta_{t} - \widehat \bbeta_{t - 1}\|_2 > \tau$}
		\UNTIL {$\|\widehat \bbeta_{t} - \widehat \bbeta_{t - 1}\|_2 \le \tau$}
		\STATE $\widehat \bbeta^{\mathrm{iht}} \leftarrow \widehat \bbeta_t$
		\STATE $\widehat \cS^{\mathrm{iht}} = \cT_{\mathrm{abs}}\bigl(\widehat\bbeta^{\mathrm{iht}}, \min(\widehat s, \pi)\bigr) \bigcup \cT_{\mathrm{abs}}\bigl(\nabla\cL\bigl(\widehat\bbeta^{\mathrm{iht}}\bigr), \max(0, \widehat s - \pi)\bigr)$
		\STATE \textbf{{Output}}: $\widehat \cS^{\mathrm{iht}}$.
	\end{algorithmic}
\end{algorithm}

\subsection{TPR guarantees of IHT}

In this section, we establish the TPR guarantee of the iterates of IHT. Let $\cL(\bbeta) = n ^ {-1}\sum_{i = 1} ^ n (y_i - \bx_i ^ \top\bbeta) ^ 2$. Define $L := \max_{|\cS|\leq 2\pi+l}\lambda_{\max}(\hat\bSigma_{\cS\cS}), \allowbreak\alpha := \min_{|\cS|\leq 2\pi+s}\lambda_{\min}(\hat\bSigma_{\cS\cS}) \text{ and } \kappa:=L / \alpha$, where $\pi$ and $l$ are the projection size and expansion size in IHT. We first present a proposition on the optimization error rate of IHT \citep[][Theorem 4]{JTK14} that serves as the backbone of our TPR analysis. 
\begin{prop}
	\label{prop:iht_opt}
	Choose $l\geq s$ and $\pi\geq 4\kappa^2l + s - l \ge 4\kappa ^ 2s$ in IHT. Denote the $t$th iteration of IHT by $\widehat \bbeta ^ {\iht}_t$. There exists a universal constant $C$ such that for any $\epsilon > 0$ and any $t  \ge C\kappa\log\frac{\cL(\widehat\bbeta^{\iht}_0)}{\epsilon}$, we have that $\cL\bigl(\widehat \bbeta ^ {\iht}_t\bigr) - \cL\bigl(\widehat\bbeta ^ {\mathrm{best}}(s)\bigr) \le \epsilon$.
\end{prop}

Proposition \ref{prop:iht_opt} shows that as long as the restricted condition number $\kappa$ is well above from zero, IHT can fast approximate the minimum objective of the best-$s$-subset problem \eqref{eq:best_subset} with $\pi$ variables. Combining this optimization guarantee with Theorem \ref{thm:power}, we can establish the following TPR guarantee of IHT.
			
\begin{thm}
	\label{thm:iht}
	Suppose that $p \ge 3$ and that the design is fixed. For any $\pi > s$, let the identifiability margin $\tau_*(\pi, \delta)$ be defined as in Theorem \ref{thm:power}. Choose the same $\pi$ and $l$ as in Proposition \ref{prop:iht_opt}. Then there exist universal constants $C_1, C_2$ such that for any $\xi > C_1, \delta \in (0, 1]$ and $0\le  \eta < 1$, whenever
	\beq
		\label{eq:min_tau_condition_cor1}
		\tau_*(\pi, \delta) \geq \biggl(\frac{4\xi}{1 - \eta}\biggr)^2\frac{\sigma^2\log p}{n},
	\eeq
 	we have that
	\[
		\PP\bigl(\tpr(\widehat\bbeta ^ {\iht}_t) \ge 1-\delta\bigr) \ge 1 - 8sp^{- (C_1^{-1}\xi - 1)}
	\]
	for any $t\geq C_2\kappa\log\frac{\cL(\widehat\bbeta^{\iht}_0)}{n\eta\tau_*(\pi, \delta)}$. In particular, when \eqref{eq:min_tau_condition_cor1} holds for $\delta < s ^ {-1}$, we have that
	\[
		\PP\bigl(\tpr(\widehat\bbeta ^ {\iht}_t)  = 1\bigr) \ge 1 - 8sp^{- (C_1^{-1}\xi - 1)}
	\]
	for any $t\geq C_2\kappa\log\frac{\cL(\widehat\bbeta^{\iht}_0)}{n\eta\tau_*(\pi, \delta)}$.
\end{thm}

\begin{rem}
    We emphasize that the restricted condition number $\kappa$ only affects the convergence speed of the IHT algorithm and has no impact on the TPR guarantee that IHT can ultimately achieve after sufficiently many iterations. In other words, $\kappa > 0$ is an algorithmic rather than statistical requirement, and $\tau_*(\pi, \delta)$ is still the underpinning quantity that determines if IHT is able to achieve sure screening. 
\end{rem}

Given the sure screening property and the sparsity level of $\widehat \bbeta ^ {\iht}_t$, one can compute the BSS problem \eqref{eq:best_subset} on $\supp(\widehat\bbeta^{\iht}_t)$ to further enhance the quality of model selection. For any sparsity estimate $\hat s$, define $\widetilde \bbeta ^ {\iht}_t(\hat s)$ to be the solution of the best-$\hat s$ subset selection on the support by $\widehat\bbeta^{\iht}_t$, i.e.,
\[
\widetilde \bbeta ^ {\iht}_t(\hat s) := \argmin_{\substack{\bbeta \in \R^p, \|\bbeta\|_0 \le \hat s, \\ \supp(\bbeta) \subset \supp( \widehat \bbeta ^ {\iht}_t)}} \cL(\bbeta).
\]
The following corollary shows that should the true sparsity is known, the resulting two-step procedure is able to recover exactly the true model with high probability.

\begin{cor}
	\label{cor:two_stage}
	Choose $l \ge s$ and $\pi \ge 4\kappa ^ 2l$ in IHT. Under the same assumptions as in Theorem \ref{thm:iht}, there exist universal constants $C_1, C_2$ such that for any $\xi > C_1$ and $0\le  \eta < 1$, whenever $ \tau_* (\pi, \delta) \ge 16\xi ^ 2\sigma ^ 2 \log (p) / \{(1 - \eta) ^ 2n\}$ for some $\delta < s ^ {-1}$, we have that
	\[
		\PP\biggl\{\supp\{\widetilde \bbeta ^ {\iht}_t(s)\} = \cS^*, \forall t\geq C_1\kappa\log\biggl(\frac{\cL(\widehat\bbeta_0^{\iht})}{n\eta\tau_*(\pi, \delta)}\biggr)\biggr\} \ge 1 - 8sp^{- (C_2^{-1}\xi - 1)}.
	\]
\end{cor}

{
\begin{rem}
Two-stage and multi-stage screening procedures are common variable seletion strategies for high-dimensional sparse linear models \citep{wasserman2009high, ji2012ups, wang2020bridge}. \cite{wasserman2009high} propose to first fit multiple candidate models on the solution path of LASSO, then select one model by cross validation and finally use hypothesis testing to eliminate some variables. Under the sparse Riesz condition, they show that the proposed approach is able to achieve model consistency. \cite{ji2012ups} propose a screen and clean approach called Univariate Penalization Screening (UPS) and show that it delivers optimal asymptotic rate of Hamming risk under sparse correlation between covariates. \cite{wang2020bridge} consider a two-stage approach that first obtains the bridge estimators of the regression coefficients and then thresholds the estimators to select important variables. They evaluate asymptotic FDP and TPR with precise constants under independent Gaussian design. The difference between our result and these is that we are mainly interested in model selection under highly correlated design. Our take-home message is that regardless of how large $\kappa$ is, our two-stage approach can achieve model consistency with high probability given the true sparsity and a sufficient number of iterations, which illustrates the robustness of the approach against design dependence. 
% Among them, the most relevant work to ours is \cite{wang2020bridge}, which consider two-stage bridge estimators under various signal-to-noise ratios and independent design ($X_{ij} \overset{i.i.d.}{\sim} \mathcal N(0, 1/n)$). 
\end{rem}
}

\section{Simulation study} \label{sec:sim}

The goal of this section is to compare the TPR-FDR curves corresponding to the solution paths of IHT and other competing methods on synthetic datasets. An ideal model selector should exhibit high TPR and low FDR once configured appropriately, yielding a $\Gamma$-shaped TPR-FDR curve. We consider the following three competing methods:
\begin{itemize}
		\item Sure Independence Screening (SIS, \cite{FLv08}): SIS selects the variables that have top marginal correlation with the response. It is essentially the very first iteration of IHT with zero initialization and standardized design.
		
		\item LASSO: LASSO chooses $p_{\lambda}(|\beta|) = \lambda |\beta|$ in \eqref{eq-Mest}.
		
		\item SCAD: SCAD chooses $p_{\lambda}$ in \eqref{eq-Mest} satifisfying that
		\[
			p'_{\lambda}(|\beta|) = \lambda\biggl\{1_{\{|\beta| \le \lambda\}} + \frac{(a\lambda - |\beta|)_+}{(a - 1)\lambda}1_{|\beta| > \lambda}\biggr\}.
		\]
	\end{itemize}
In IHT, we choose the projection size $\pi$ to be $50$ or $100$. Then we plot the TPR against the FDR of $\widehat \bbeta^{\iht}(\hat s)$ as $\hat s$ varies from $1$ to $p$. As for LASSO and SCAD, we compute and present their TPR and FDR as $\lambda$ goes through a properly predefined sequence. Moreover, we point out the average FDR and TPR of LASSO and SCAD with $\lambda$ maximizing the 10-fold cross validation score. Columns of $\bX$ are always standardized before being fed to the algorithms.

We generate the data as follows:
\begin{enumerate}
	\item $p=1,000$, $s=50$, $\cS^* = [s]$ and $n=\lceil{2s\log p}\rceil$; \
	\item $\beta^*_j = 0$ for $j \in (\cS^*)^c$, and $\{(\beta^*_j / \beta_{\min}) - 1\}_{j \in \cS ^ *} \overset{\mathrm{i.i.d.}}{\sim} \chi_1^2$, where $\beta_{\mathrm{min}} = 0.1$;
	\item $\{\bx_i\}_{i \in [n]}\overset{\mathrm{i.i.d.}}{\sim}\mathcal N(\bzero,\bSigma)$ and $\{\epsilon_i\}_{i \in [n]}\overset{\mathrm{i.i.d.}}{\sim}\mathcal N(0,\sigma^2)$, where $\bSigma$ and $\sigma$ are specified in the subsequent subsections.
\end{enumerate}
We consider three possible setups of $\bSigma$ for comparison: covariance with exponential decay, equicorrelation and a factor model.

\subsection{Covariance with exponential decay}

Here for $i, j\in[p]$, we set $\bSigma_{i,j}=q^{|i-j|}$ where we choose $q=0,0.5,0.8$. We consider two noise levels: $\sigma=0.3$ or 0.6. We present the TPR-FDR curves of the aforementioned selection methods in Fig. \ref{exp}. We have the following observations: 

\begin{enumerate}
		\item[(i)] In the low-noise setting ($\sigma = 0.3$), the IHT algorithm yields nearly a $\Gamma$-shaped TPR-FDR curve regardless of the choice of $\pi = 50$ or $100$. In contrast, as the design becomes more and more correlated ($q$ increases), the TPR-FDR curves of LASSO and SCAD gradually deviate from the $\Gamma$ shape. SIS performs poorly, as it ignores the correlation of the designs. 
		
		\item[(ii)] In the high-noise setting ($\sigma = 0.6$), all the investigated approaches perform worse than in the low-noise setting. Nevertheless, as clearly illustrated by Figure \ref{exp}(f), IHT yields reasonably high TPR when FDR is below $0.1$, while neither LASSO, SCAD nor SIS has TPR higher than $0.2$ when their FDR is below 0.1. This suggests that IHT is better than LASSO, SCAD and SIS in terms of controlling FDR in the presence of strong collinearity of design and high noise level. 
		
		\item[(iii)] The CV-tuned LASSO and SCAD have high FDR, suggesting that they tend to select dense models to achieve good prediction performance.
	\end{enumerate}

\begin{figure}[bht]
\centerline{\includegraphics[width=16cm]{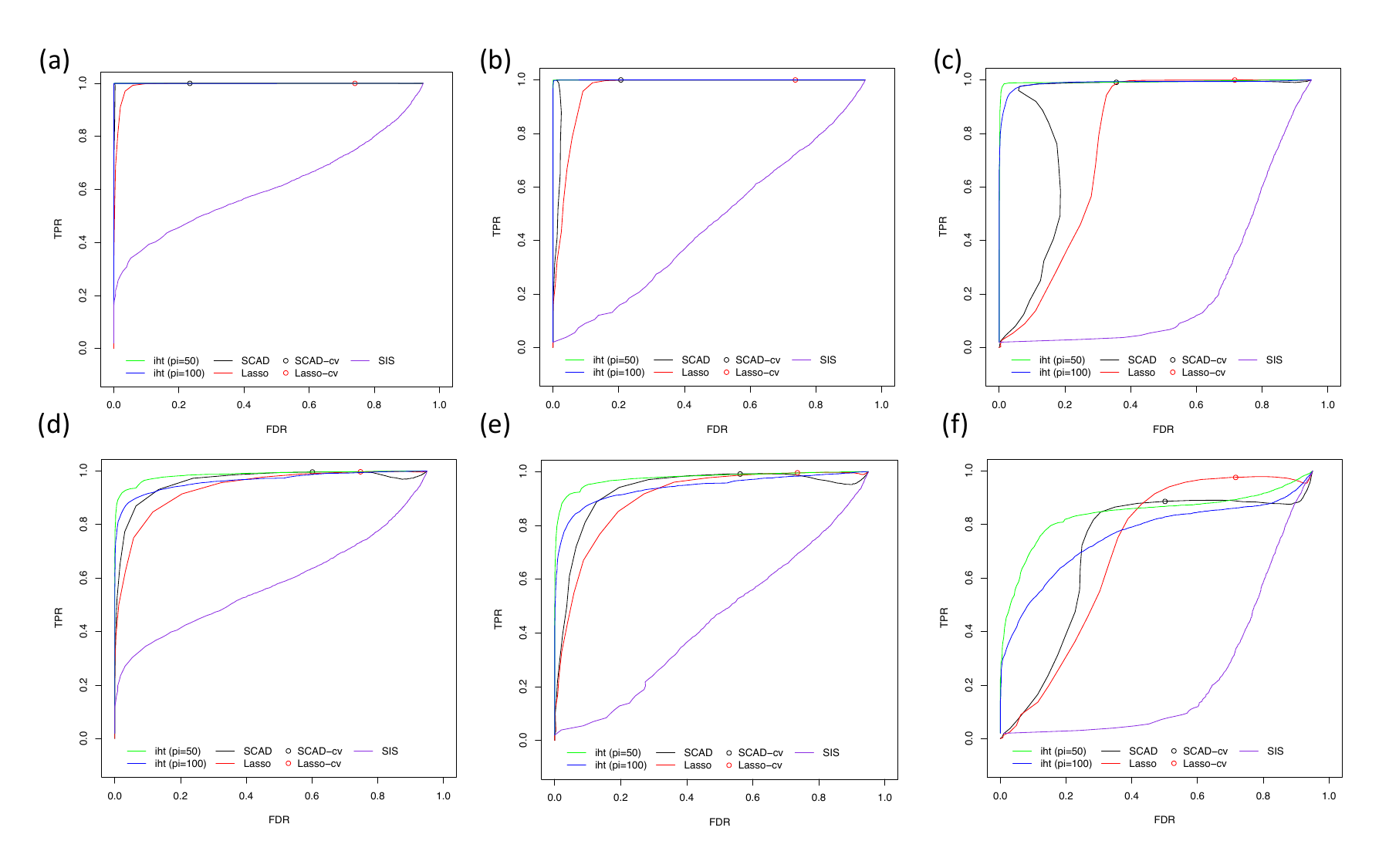}}
\caption{Setting 1: Covariance with exponential decay ($\bSigma_{i,j}=q^{|i-j|}$ where $q=0,0.5$ and $0.8$). In the first and second rows, the noise level is set as $\sigma=0.3$ and $0.6$ respectively; The first, second and third columns correspond to $q=0,0.5$ and $0.8$ respectively. The black and red lines represent the TPR-FDR curves of $\widehat \bbeta ^ {\pen}$ as $\lambda$ varies, while the green and blue lines represent the TPR-FDR curves of $\{\widehat \bbeta^{\iht}(\widehat s)\}_{\widehat s=1}^p$ with projection size $\pi=50$ and $100$ respectively as $\hat s$ varies. The blueviolet curve represents the TPR-FDR curve of SIS as the selected model size varies. The dots indicate the TPR and FDR of $\widehat \bbeta^{\pen}$ with $\lambda$ chosen by 10-fold cross validation. }
\label{exp}
\end{figure}

\subsection{Equicorrelation}
Here we set $\bSigma$ as follows: $\bSigma_{i,j}=1$ if $i=j$, and $\bSigma_{i,j}=q$ otherwise. We again consider two noise levels: $\sigma=0.3$ or $0.6$, and illustrate the TPR and FDR of all the methods in Figure \ref{const}. We have the following two observations: 

\begin{enumerate}
		\item[(i)] The performance gap between IHT and other methods further widens compared with Section 4.1: IHT still yields a nearly $\Gamma$-shaped TPR-FDR curve, while neither LASSO, SCAD nor SIS does so. In particular, in Figure \ref{const}(d), IHT yields reasonably high TPR when its FDR is below $0.1$, while LASSO, SCAD and SIS have nearly zero TPR when their FDR is below 0.1. This further suggests that IHT is substantially more robust to collinearity of the design than LASSO, SCAD and SIS in terms of model selection. 
		
		\item[(ii)] Blessed by the nonconvexity of the penalty, SCAD is able to correct its selected model by replacing spurious variables with correct ones when $\lambda$ decreases within a certain phase. This is the reason that the TPR-FDR curve of SCAD can pivot in the north-west direction to the oracle point ($\mathrm{TPR} = 1$, $\mathrm{FDR} = 0$), and that the CV-tuned SCAD estimator has lower FDR than the CV-tuned LASSO estimator. 
\end{enumerate}

\begin{figure}[bht]
\centerline{\includegraphics[width=12cm]{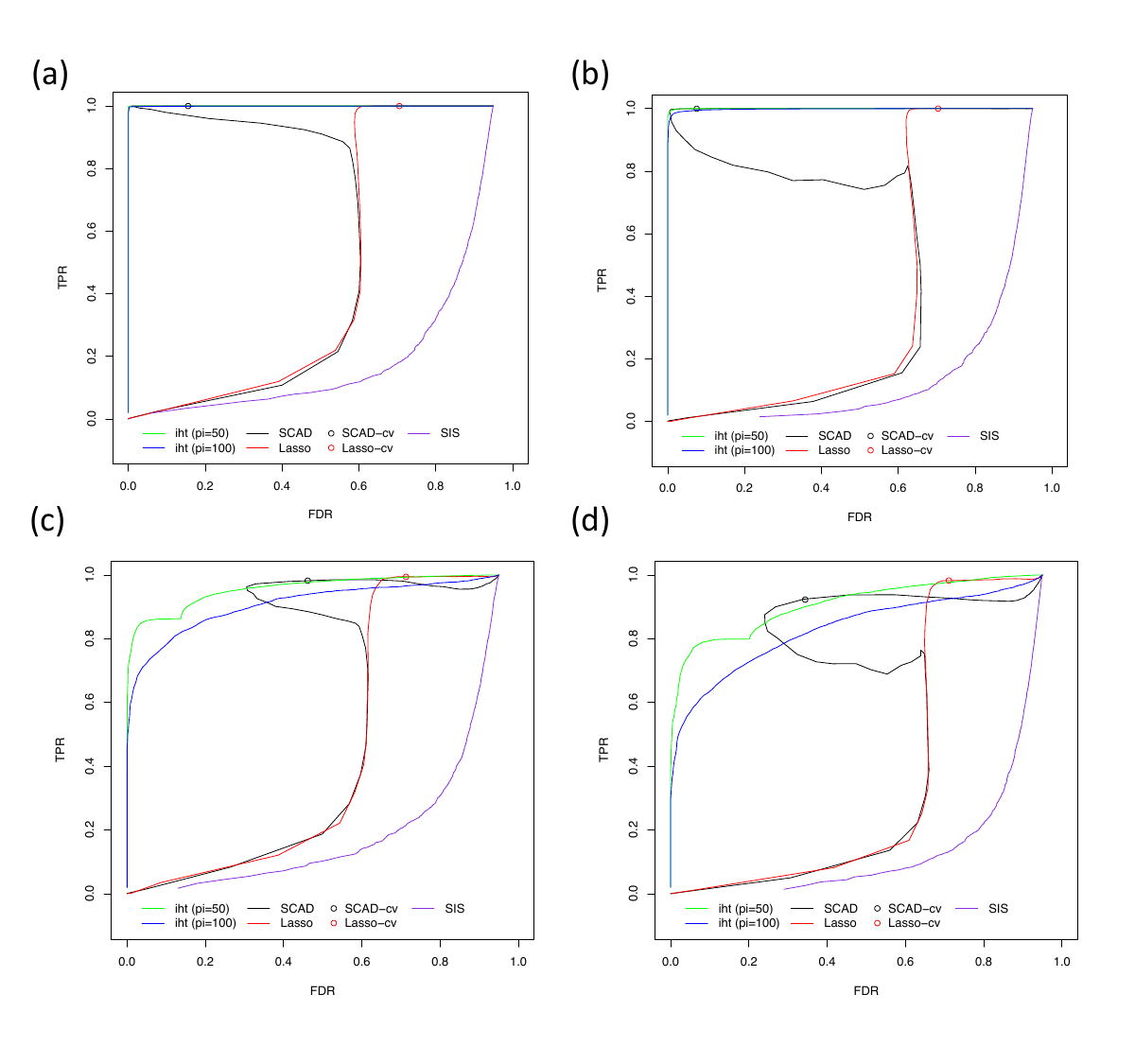}}
\caption{Setting 2: Constant correlation model ($\bSigma_{i,j}=1$ if $i=j$, and $\bSigma_{i,j}=q$ otherwise). In the first and second rows, the noise level is set as $\sigma=0.3$ and $0.6$ respectively. The first and second columns correspond to $q=0.5$ and $0.8$ respectively. 
%	\jf{Q: right caption?} 
The black and red lines represent the TPR-FDR curves of $\widehat \bbeta ^ {\pen}$ as $\lambda$ varies, while the green and blue lines represent the TPR-FDR curves of $\{\widehat \bbeta^{\iht}(\widehat s)\}_{\widehat s=1}^p$ with projection size $\pi=50$ and $100$ respectively as $\hat s$ varies. The dots indicate the TPR and FDR of $\widehat \bbeta^{\pen}$ with $\lambda$ chosen by 10-fold cross validation. The blueviolet curve represents the TPR-FDR curve of SIS.}
\label{const}
\end{figure}

\subsection{Factor model}

Here we let $\bSigma=\bSigma_b+\bSigma_u$, where $\bSigma_u=\bI$ and $\bSigma_b=\bV\bLambda_0
\bV^\top$ contains the spiky part of the covariance structure. Here we let $\bV\in \cO_{p, K}$, where
$$
\cO_{p, K}=\left\{ \bU\in\R^{p\times K}: \bU^\top\bU=\bI_{K\times K}  \right\}.
$$
We let $\bLambda_0\in\R^{K\times K}$ be a diagonal matrix consisting of the $K$ spiky eigenvalues of $\bSigma_b$. We let $K=2$ and consider the following two cases: $\bLambda_0=\diag(2p,p)$ and $\diag(2\sqrt{p},\sqrt{p})$. We present the TPR-FDR curves in Fig. \ref{fm}. Similarly to the previous two cases, IHT is still the best among all the investigated methods in terms of FDR and TPR, especially when the covariance structure is more spiky ($\bLambda_0 = \diag(2p, p)$). Besides, as $\lambda$ decreases from its largest value, the TPR-FDR curve of SCAD first goes east and then pivots sharply to the west to follow the curve of IHT, thanks to its nonconvex penalty. 

\begin{figure}[bht]
\centerline{\includegraphics[width=12cm]{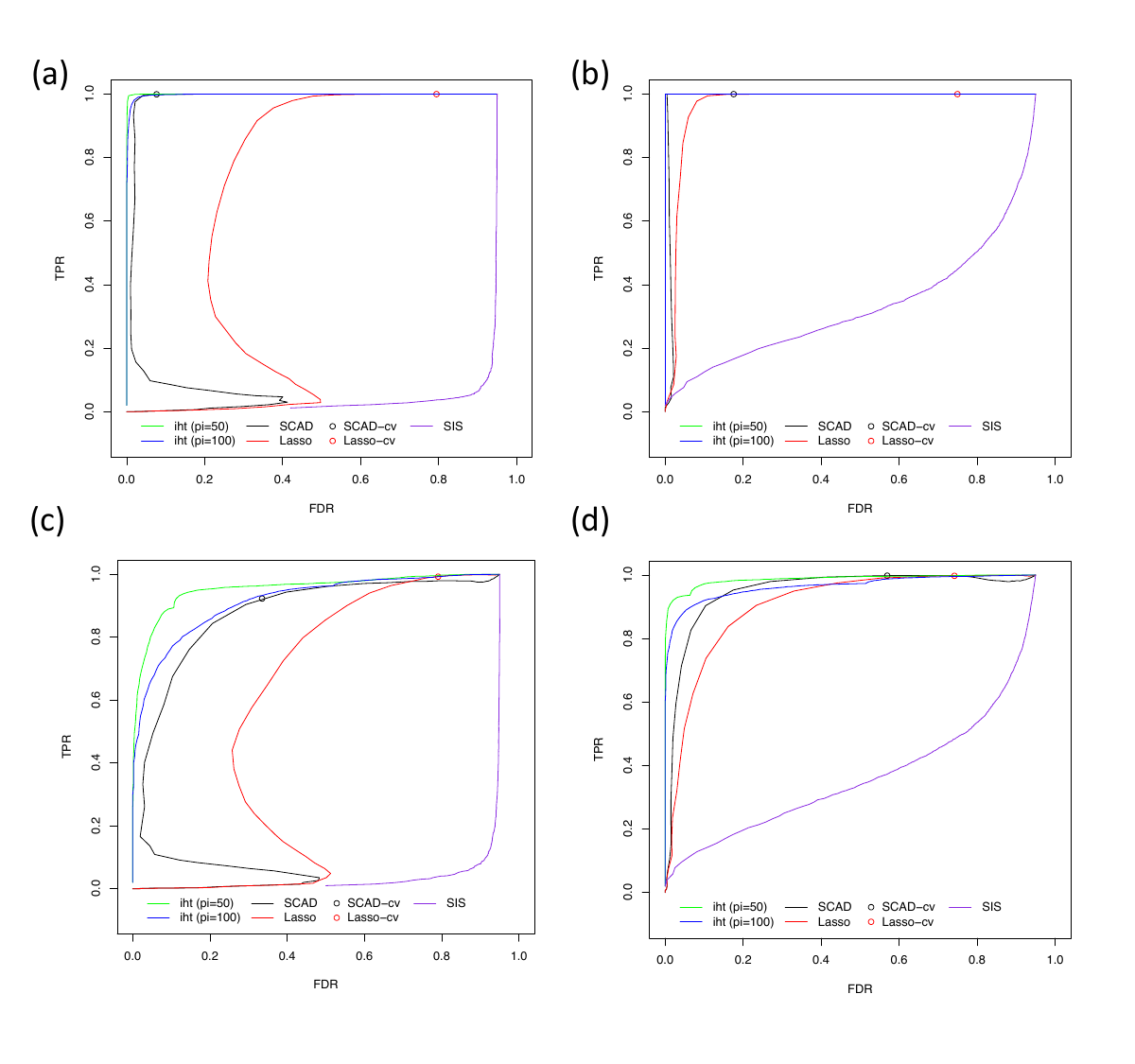}}
\caption{Setting3: Factor model. In the first and second rows, the noise level is set as $\sigma=0.3$ and $0.6$ respectively; The first and second columns correspond to $\bLambda_0=\diag(2p,p)$ and $\diag(2\sqrt{p},\sqrt{p})$ respectively. The black and red lines represent the TPR-FDR curves of $\widehat \bbeta ^ {\pen}$ as $\lambda$ varies, while the green and blue lines represent the TPR-FDR curves of $\{\widehat \bbeta^{\iht}(\widehat s)\}_{\widehat s=1}^p$ with projection size $\pi=50$ and $100$ respectively as $\hat s$ varies. The dots indicate the TPR and FDR of $\widehat \bbeta^{\pen}$ with $\lambda$ chosen by 10-fold cross validation. The blueviolet curve represents the TPR-FDR curve of SIS.}
\label{fm}
\end{figure}

\section{Real datasets} \label{sec:real}
\subsection{The Diabetes dataset}
We first consider the Diabetes Dataset that was studied in \cite{BTI04} and \cite{BKM16}. The response of interest is a quantitative measure of disease progression one year after baseline, and the predictors include ten baseline variables (age, sex, body-mass index, etc) as well as their quadratic terms. The total sample size $n=442$, and the dimension $p=64$. All feature columns are centered and normalized such that their $L_2$-norms are ones.

To compare LASSO, SCAD, SIS and IHT, we randomly divide the dataset into a training set ($80\%$ observations) and a testing set ($20\%$ observations). Then we apply these four algorithms to the training set with tuning parameters chosen by cross validation. We investigate the testing performance as well as the size of the trained model. For SIS, the tuning parameter is the number of features selected according to marginal correlation, and out-of-sample $R^2$ is calculated using the least squares refitted model on the top features.
%\begin{itemize}
%\item The model selected by each algorithm (model size, important features);
%\item The prediction performance of the trained model on the testing set.
%\end{itemize}
The results are shown in table \ref{table:diabetes}. As we can see, IHT selects a much sparser model than both LASSO and SCAD, while achieving a similar out-of-sample $R^2$ as LASSO and SCAD. Besides, IHT agrees with LASSO and SCAD on the most important features: \texttt{bmi} (body mass index), \texttt{ltg}, \texttt{map} (two blood serum measurements), \texttt{age.sex} (interaction between the variables \texttt{age} and \texttt{sex}), \texttt{hdl} (a blood serum measurement) and \texttt{sex} (sex). SIS obtains a worse $R^2$, and the top selected features are different.
%One reason might be that it cannot deal with highly correlated features.

\begin{table}
	\caption{Model selection and prediction of Lasso, SCAD, IHT and SIS on the Diabetes dataset \citep{BTI04}. The column ``$R^2$'' represents out-of-sample $R^2$ on the test dataset; The column ``Model Size'' represents the number of features selected by the trained model; the ``Most Significant Features'' shows the top 6 features corresponding entries with the highest $p$ values in the refitted coefficients. The meanings of the features shown here are explained in the main text.}	
 \centering
 \begin{tabular}{cccc}
 \hline\hline
  & $R ^ 2$ & Model Size & Most Significant Features (top 6) \\
 \hline\hline
 LASSO & 0.537 &14 & \texttt{bmi, ltg, map, age.sex, hdl, sex} \\
 \hline
 SCAD & 0.562 & 16 & \texttt{bmi, ltg, map, age.sex, hdl, sex} \\
 \hline
 IHT & 0.554 & 6 & \texttt{bmi, ltg, map, age.sex, hdl, sex} \\
  \hline
  SIS & 0.517 & 9 & \texttt{ltg, bmi, map, bmi$^2$, tc, glu} \\
  \hline
\end{tabular}

\label{table:diabetes}
\end{table}

Moreover, we assess all the four methods with additional artificial noise features. Specifically, we add $p_n$ (ranging from 100 to 500) noise features that are highly correlated with each other but independent of the original features. The noise features are Gaussian with mean 0 and covariance matrix $\bSigma_{p_n}=0.5 \bI_{p_n} +0.5 {\bfm 1}_{p_n}{\bfm 1}_{p_n}^\top\in \RR^{p_n\times p_n}$. All the features are standardized before being fed into the algorithms. After that, we randomly divide the dataset into a training set and a testing set as before. We then perform the variable selection procedures and examine the model size and the number of noise variables that are falsely selected. The results are shown in {Fig. \ref{fig-pn_modelsize}}.
%We observe that as the number of noise features increases, none of the algorithms yields significantly different out-of-sample $R^ 2$. \zzw{Where can I see the $R ^ 2$?}
We can see that as the number of noise features increases, LASSO and SCAD select larger models with more noise variables. In particular, when $p_n = 500$, around half of the features selected by LASSO are artificial noise features. In contrast, IHT always selects a small model with a tiny fraction of noise variables. SIS also selects a simple model with few noise features consistently, because the added noise features are independent of the original data. 

\begin{figure}[htb]
\centerline{\includegraphics[width=9cm]{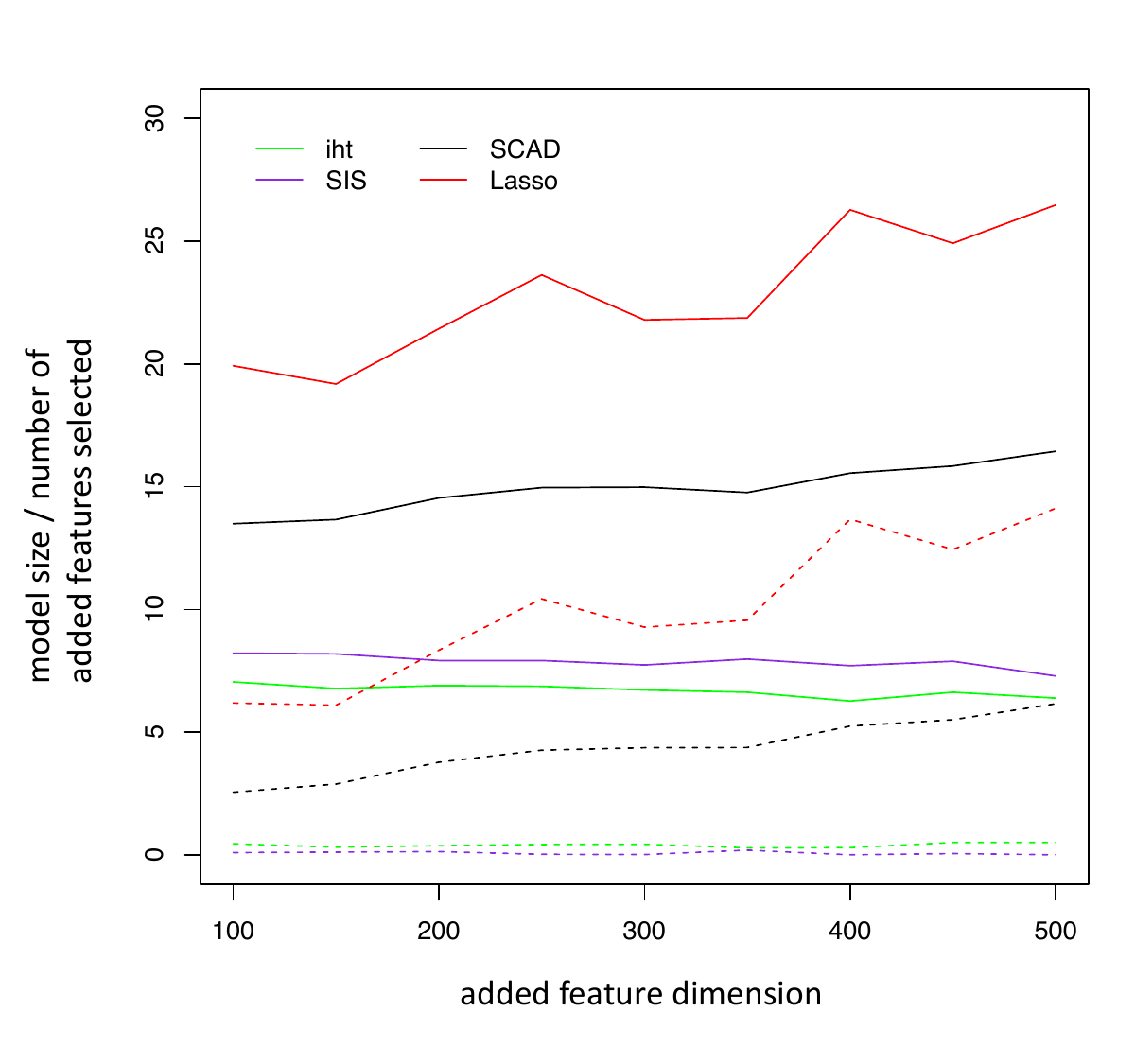}}
\caption{Model selection and prediction of Lasso, SCAD, IHT and SIS on the Diabetes dataset \citep{BTI04} with additional noise features. The noise features are added independently of the original features and follow $\mathcal N(0, \bSigma_{p_{n}})$, where $\bSigma_{p_n}=0.5 \bI_{p_n} +0.5 {\bfm 1}_{p_n}{\bfm 1}_{p_n}^\top\in \RR^{p_n\times p_n}$ has 1's in all its diagonal entries and 0.5's in its off-diagonal entries. All the algorithms are evaluated through 100 independent generations of random noise features with $p_n = 100, 150, \ldots, 500$.
The solid lines represent the average model size, while the dashed lines represent the number of artificial noise features that are mistakenly selected.}
\label{fig-pn_modelsize}
\end{figure}

\subsection{The Monthly Macroeconomic Dataset}
We now turn to a macroeconomic dataset extracted from the FRED-MD database \citep{MN16}. The dataset contains monthly observations of 129 macroeconomic variables covering aspects such as labor market, housing,  consumption, money and credit, interest and exchange rates, prices, stock market, etc. Our primary goal is to conduct association studies and find out how these variables are related to each other. In particular, we study how \emph{unemployment rate} and \emph{consumer price index} are associated with the other macroeconomic variables. Towards this end, we extract observations from January 1980 to November 2018 and use the last ten years' data as the testing data and the rest as the training data. For each target variable, in pursuit of a meaningful model, we delete the columns that are related with it in a striaghtforward and trivial manner. For instance, when predicting the unemployment rate, we delete the columns such as the number of civilians unemployed for fewer than 5 weeks, number of civilians unemployed for 5 to 14 weeks, number of civilians unemployed for 15 to 26 weeks, etc. Then, we apply the four algorithms assessed in the previous subsection with tuning parameters chosen by cross validation. As in the analysis of the Diabetes Dataset, we assess both the prediction performance and the size of the selected model.

Tables \ref{table:UNRATE} and \ref{table:CPIAUCSL} show the output model size, top five important features as well as the out-of-sample $R^2$ of the four methods when we predict the unemployment rate and CPI.  As different models yield different model sizes, to make more fair comparisons, we additionally report the out-of-sample $R^2$ of the refitted least squares model using the 10 most significant variables in each model (If model size is less than 10, then refitting is done using all selected variables). In both tasks, IHT achieves a similar $R ^ 2$ as LASSO and SCAD with a much smaller model. Moreover, the ten most significant features selected by IHT also achieve better prediction performance than those selected by the other methods. In addition, the performance of SIS is relatively unstable: while SIS has similar $R ^ 2$ as the other approaches in predicting CPI, it performs much more poorly than the others in predicting the unemployment rate. 
%when predicting the unemployment rate, SIS has significantly less $R^2$ than the other methods; when predicting CPI, SIS includes the features from the same sector that none of the other three algorithms regard as important. 
{To save the space in the main text, we relegate the meanings of the features to Section \ref{section:macro-var} in the Appendix.}

\begin{table}
\caption{\label{table:UNRATE}Model selection and prediction of Lasso, SCAD, IHT and SIS on the macroeconomic dataset \citep{MN16} for unemployment rate association studies. The column ``$R^2$'' represents out-of-sample $R^2$ on the test dataset; the column ``Model Size'' represents the number of features selected by the trained model; the column ``$R^2_{(10)}$'' represents the out-of-sample $R^2$ of the refitted least squares model using the 10 most important variables in each model (If model size is less than 10, then refitting is done using all selected variables; the ``Most Significant Features'' column gives the top 5 features corresponding to entries with the least p-values in the refitted coefficients. {To save the space in the main text, we relegate the meanings of the features to Section \ref{section:macro-var} in the Appendix.}}	
 \begin{tabular}{ccccc}
 \hline\hline
  & $R^2$ & Model Size & $R^2_{(10)}$& Most Significant Features (top 5) \\
 \hline\hline
 LASSO & 0.517 & 40 & 0.462 & \texttt{HWIURATIO, HWI, COMPAPFFx, M1SL, UEMPMEAN} \\
 \hline
 SCAD & 0.422 & 7 & 0.362 & \texttt{HWIURATIO, HWI, DMANEMP, PAYEMS, UEMPMEAN} \\
 \hline
  IHT & 0.470 & 4 & 0.470 & \texttt{HWIURATIO, HWI, DMANEMP, PAYEMS} \\
  \hline
  SIS & 0.171 & 10 & 0.171 & \texttt{HWIURATIO, IPDMAT, IPMANSICS, INDPRO, PAYEMS}\\
  \hline
\end{tabular}
\end{table}

\begin{table}
\caption{\label{table:CPIAUCSL}	Model selection and prediction of Lasso, SCAD, IHT and SIS on the macroeconomic dataset \citep{MN16} for CPI association studies. The column ``$R^2$'' represents out-of-sample $R^2$ on the test dataset; the column ``Model Size'' represents the number of features selected by the trained model; the column ``$R^2_{(10)}$'' represents the out-of-sample $R^2$ of the refitted least squares model using the 10 most important variables in each model (If model size is less than 10, then refitting is done using all selected variables); the ``Most Significant Features'' shows the top 5 features corresponding entries with the highest $p$ values in the refitted coefficients. The meanings of the features shown here are explained in the main text.}
\centering
 \begin{tabular}{ccccp{.5\textwidth}}
 \hline\hline
  & $R^2$ & Model Size &$R^2_{(10)}$& \hspace{1cm} Most Significant Features (top 5) \\
 \hline\hline
 LASSO & 0.902 & 20 &0.876& \texttt{DNDGRG3M086SBEA, PCEPI, FEDFUNDS, NDMANEMP, BUSINVx} \\
 \hline
 SCAD & 0.909 & 15 &0.891& \texttt{DNDGRG3M086SBEA, PCEPI, FEDFUNDS, NDMANEMP, WPSID61} \\
 \hline
 IHT & 0.905 & 2& 0.905 & \texttt{DNDGRG3M086SBEA, PCEPI} \\
  \hline
 SIS & 0.903 & 6 &0.903& \texttt{DNDGRG3M086SBEA, PCEPI, WPSID61, WPSID62, WPSFD49207} \\
  \hline
\end{tabular}

\end{table}

Similarly to the diabetes dataset, we further explore the variable selection properties of all the four algorithms by incorporating noise features. Specifically, we generate $p_n$ ($p_n=10, 30, 50$) spurious Gaussian features with mean 0 and covariance matrix $\bSigma_{p_n}=0.5 \bI_{p_n} +0.5  {\bfm 1}_{p_n}{\bfm 1}_{p_n}^\top$ (independently of the original features). All features are standardized before being fed into the algorithms. Then, after randomly dividing the dataset into a training set and a testing set, we apply all the four algorithms and examine their out-of-sample $R^2$, model size, and number of noise variables that are selected into the model. The results are shown in Tables \ref{table:UNRATE_noise} and \ref{table:CPI_noise}. We can observe that as the number of noise variables increases, the out-of-sample $R^2$ and the model size remain stable for all the algorithms. Nevertheless, IHT and SIS consistently select very few noise variables, while LASSO and SCAD select an increasing number of noise variables. 
In particular, IHT and SIS never select any noise variables in the CPI association study. 
This in turn suggests that LASSO might select quite a number of spurious variables in the earlier study without artificial noise variables.

\begin{table}
\caption{\label{table:UNRATE_noise}Model selection and prediction of Lasso, SCAD, IHT and SIS on the macroeconomic dataset \cite{MN16} with additional noise features for unemployment rate association studies. The noise features are added independently of the original features, and are generated with the distribution $\mathcal N(0, \bSigma_{p_{n}})$, where $\bSigma_{p_n}=0.5\bI_{p_n} +0.5 {\bfm 1}_{p_n}{\bfm 1}_{p_n}^\top\in \RR^{p_n\times p_n}$ has 1 in all its diagonal entries and 0.5 in its off-diagonal entries. All algorithms are evaluated through 100 independent generation of random noise features with $p_n = 10, 30, 50$. The column ``$R^2$'' represents the averaged out-of-sample $R^2$ on the test dataset; The column ``Model Size'' represents the average number of features selected by the trained model; The ``Noise Variables Selected'' column gives the average number of noise features that are selected into the model. The associated standard errors are put in the subscript.}

	\centering
 \begin{tabular}{ccccc}
 \hline\hline
  & & $R ^ 2$ & Model Size & Noise Variables Selected \\
 \hline\hline
 &LASSO & $0.495_{(0.003)}$ &$49.750_{(1.106)}$ & $3.840_{(0.234)}$\\
 $p_{n}=10$&SCAD &$0.424_{(0.002)}$ &$11.310_{(0.340)}$ & $0.550_{(0.073)}$\\
 &IHT & $0.497_{(0.004)}$ & $6.880_{(0.167)}$ & $0.040_{(0.032)}$ \\
  &SIS & $0.148_{(0.003)}$ & $8.360_{(0.198)}$ & $0_{(0)}$ \\
 \hline
 &LASSO & $0.550_{(0.003)}$ &$47.820_{(1.360)}$ & $7.610_{(0.488)}$\\
 $p_{n}=30$&SCAD &$0.426_{(0.001)}$ &$10.960_{(0.283)}$ & $0.830_{(0.102)}$\\
 &IHT & $0.479_{(0.006)}$ & $6.320_{(0.109)}$ & $0.040_{(0.024)}$ \\
 &SIS & $0.145_{(0.003)}$ & $8.080_{(0.201)}$ & $0_{(0)}$ \\
 \hline
&LASSO & $0.510_{(0.003)}$ &$49.630_{(1.679)}$ & $11.850_{(0.806)}$\\
 $p_{n}=50$&SCAD &$0.425_{(0.002)}$ &$12.730_{(0.433)}$ & $1.910_{(0.180)}$\\
 &IHT & $0.468_{(0.007)}$ & $6.120_{(0.069)}$ & $0.040_{(0.020)}$ \\
 &SIS & $0.150_{(0.003)}$ & $8.480_{(0.195)}$ & $0_{(0)}$ \\
 \hline
\end{tabular}
\end{table}

\begin{table}
\caption{\label{table:CPI_noise}Model selection and prediction of Lasso, SCAD, IHT and SIS on the macroeconomic dataset \cite{MN16} with additional noise features for CPI association studies. The noise features are added independently of the original features, and are generated with the distribution $\mathcal N(0, \bSigma_{p_{n}})$, where $\bSigma_{p_n}=0.5 \cdot \bI_{p_n} +0.5 \cdot {\bfm 1}_{p_n}{\bfm 1}_{p_n}^\top\in \RR^{p_n\times p_n}$ has 1 in all its diagonal entries and 0.5 in its off-diagonal entries. All algorithms are evaluated through 100 independent generation of random noise features with $p_n = 10, 30, 50$. The column ``$R^2$'' represents the averaged out-of-sample $R^2$ on the test dataset; The column ``Model Size'' represents the average number of features selected by the trained model; The ``Noise Variables Selected'' column gives the average number of noise features that are selected into the model. The associated standard errors are put in the subscript.}

	\centering
 \begin{tabular}{ccccc}
 \hline\hline
  & & $R ^ 2$ & Model Size & Noise Variables Selected \\
 \hline\hline
 &LASSO & $0.904_{(5e-4)}$ &$20.900_{(0.486)}$ & $1.400_{(0.130)}$\\
 $p_{n}=10$&SCAD &$0.908_{(1e-4)}$ &$13.670_{(0.313)}$ & $0.370_{(0.065)}$\\
 &IHT & $0.905_{(0)}$ & $2.000_{(0)}$ & $0_{(0)}$ \\
 &SIS & $0.903_{(0)}$ & $6.000_{(0)}$ & $0_{(0)}$ \\
 \hline
 &LASSO & $0.905_{(6e-4)}$ & $20.870_{(0.668)}$ & $2.760_{(0.259)}$\\
 $p_{n}=30$&SCAD &$0.908_{(2e-4)}$ &$15.030_{(0.268)}$ & $1.210_{(0.109)}$\\
 &IHT & $0.905_{(0)}$ & $2.000_{(0)}$ & $0_{(0)}$ \\
 &SIS & $0.903_{(0)}$ & $6.000_{(0)}$ & $0_{(0)}$ \\
 \hline
&LASSO & $0.905_{(7e-4)}$ &$22.120_{(0.730)}$ & $3.640_{(0.331)}$\\
 $p_{n}=50$&SCAD &$0.909_{(1e-4)}$ &$13.710_{(0.334)}$ & $1.050_{(0.103)}$\\
 &IHT & $0.905_{(0)}$ & $2.000_{(0)}$ & $0_{(0)}$ \\
 &SIS & $0.903_{(0)}$ & $6.000_{(0)}$ & $0_{(0)}$ \\
 \hline
\end{tabular}

\end{table}\textit{}

To summarize, IHT yields outstanding performance on both real datasets. Compared with LASSO and SCAD, IHT yields a much simpler model and is more robust to spurious artificial features, while achieving similar out-of-sample $R ^ 2$. Compared with SIS, IHT achieves a much higher out-of-sample $R^2$.

\bibliographystyle{ims}
\bibliography{./bib}

\begin{thebibliography}{40}
\expandafter\ifx\csname natexlab\endcsname\relax\def\natexlab#1{#1}\fi
\expandafter\ifx\csname url\endcsname\relax
  \def\url#1{\texttt{#1}}\fi
\expandafter\ifx\csname urlprefix\endcsname\relax\def\urlprefix{URL }\fi

\bibitem[{Akaike(1974)}]{Aka74}
\textsc{Akaike, H.} (1974).
\newblock A new look at the statistical model identification.
\newblock \textit{IEEE Transactions on Automatic Control} \textbf{19} 716--723.

\bibitem[{Akaike(1998)}]{Aka98}
\textsc{Akaike, H.} (1998).
\newblock Information theory and an extension of the maximum likelihood
  principle.
\newblock \textit{Selected Papers of Hirotugu Akaike}  199--213.

\bibitem[{Barron et~al.(1999)Barron, Birg{\'e} and Massart}]{BBM99}
\textsc{Barron, A.}, \textsc{Birg{\'e}, L.} and \textsc{Massart, P.} (1999).
\newblock Risk bounds for model selection via penalization.
\newblock \textit{Probability Theory and Related Fields} \textbf{113} 301--413.

\bibitem[{Bertsimas et~al.(2016)Bertsimas, King and Mazumder}]{BKM16}
\textsc{Bertsimas, D.}, \textsc{King, A.} and \textsc{Mazumder, R.} (2016).
\newblock Best subset selection via a modern optimization lens.
\newblock \textit{The Annals of Statistics} \textbf{44} 813--852.

\bibitem[{Bertsimas and Van~Parys(2020)}]{bertsimas2020sparse}
\textsc{Bertsimas, D.} and \textsc{Van~Parys, B.} (2020).
\newblock Sparse high-dimensional regression: Exact scalable algorithms and
  phase transitions.
\newblock \textit{The Annals of Statistics} \textbf{48} 300--323.

\bibitem[{Blumensath and Davies(2008)}]{BDa08}
\textsc{Blumensath, T.} and \textsc{Davies, M.~E.} (2008).
\newblock Iterative thresholding for sparse approximations.
\newblock \textit{Journal of Fourier Analysis and Applications} \textbf{14}
  629--654.

\bibitem[{Blumensath and Davies(2009)}]{BDa09}
\textsc{Blumensath, T.} and \textsc{Davies, M.~E.} (2009).
\newblock Iterative hard thresholding for compressed sensing.
\newblock \textit{Applied and Computational Harmonic Analysis} \textbf{27}
  265--274.

\bibitem[{B{\"u}hlmann and Van De~Geer(2011)}]{BVa11}
\textsc{B{\"u}hlmann, P.} and \textsc{Van De~Geer, S.} (2011).
\newblock \textit{Statistics for High-Dimensional Data: Methods, Theory and
  Applications}.
\newblock Springer Science \& Business Media.

\bibitem[{Chen et~al.(1998)Chen, Donoho and Saunders}]{CDS98}
\textsc{Chen, S.~S.}, \textsc{Donoho, D.~L.} and \textsc{Saunders, M.~A.}
  (1998).
\newblock Atomic decomposition by basis pursuit.
\newblock \textit{SIAM Journal on Scientific Computing} \textbf{20} 33--61.

\bibitem[{Efron et~al.(2004)Efron, Hastie, Johnstone and Tibshirani}]{BTI04}
\textsc{Efron, B.}, \textsc{Hastie, T.}, \textsc{Johnstone, I.} and
  \textsc{Tibshirani, R.} (2004).
\newblock Least angle regression.
\newblock \textit{The Annals of Statistics} \textbf{32} 407--499.

\bibitem[{Fan and Li(2001)}]{FLi01}
\textsc{Fan, J.} and \textsc{Li, R.} (2001).
\newblock Variable selection via nonconcave penalized likelihood and its oracle
  properties.
\newblock \textit{Journal of the American Statistical Association} \textbf{96}
  1348--1360.

\bibitem[{Fan et~al.(2020)Fan, Li, Zhang and Zou}]{FLZ20}
\textsc{Fan, J.}, \textsc{Li, R.}, \textsc{Zhang, C.-H.} and \textsc{Zou, H.}
  (2020).
\newblock \textit{Statistical Foundations of Data Science}.
\newblock CRC press.

\bibitem[{Fan et~al.(2018)Fan, Liu, Sun and Zhang}]{FLS18}
\textsc{Fan, J.}, \textsc{Liu, H.}, \textsc{Sun, Q.} and \textsc{Zhang, T.}
  (2018).
\newblock {I-LAMM} for sparse learning: Simultaneous control of algorithmic
  complexity and statistical error.
\newblock \textit{The Annals of Statistics} \textbf{46} 814.

\bibitem[{Fan and Lv(2008)}]{FLv08}
\textsc{Fan, J.} and \textsc{Lv, J.} (2008).
\newblock Sure independence screening for ultrahigh dimensional feature space.
\newblock \textit{Journal of the Royal Statistical Society: Series B
  (Statistical Methodology)} \textbf{70} 849--911.

\bibitem[{Fan and Lv(2011)}]{fan2011nonconcave}
\textsc{Fan, J.} and \textsc{Lv, J.} (2011).
\newblock Nonconcave penalized likelihood with {NP}-dimensionality.
\newblock \textit{IEEE Transactions on Information Theory} \textbf{57}
  5467--5484.

\bibitem[{Fan and Peng(2004)}]{FPe04}
\textsc{Fan, J.} and \textsc{Peng, H.} (2004).
\newblock Nonconcave penalized likelihood with a diverging number of
  parameters.
\newblock \textit{The Annals of Statistics} \textbf{32} 928--961.

\bibitem[{Foster et~al.(2015)Foster, Karloff and Thaler}]{FKT15}
\textsc{Foster, D.}, \textsc{Karloff, H.} and \textsc{Thaler, J.} (2015).
\newblock Variable selection is hard.
\newblock In \textit{Conference on Learning Theory}.

\bibitem[{Hastie et~al.(2017)Hastie, Tibshirani and Tibshirani}]{HTT17}
\textsc{Hastie, T.}, \textsc{Tibshirani, R.} and \textsc{Tibshirani, R.~J.}
  (2017).
\newblock Extended comparisons of best subset selection, forward stepwise
  selection, and the lasso.
\newblock \textit{arXiv preprint arXiv:1707.08692} .

\bibitem[{Jain et~al.(2014)Jain, Tewari and Kar}]{JTK14}
\textsc{Jain, P.}, \textsc{Tewari, A.} and \textsc{Kar, P.} (2014).
\newblock On iterative hard thresholding methods for high-dimensional
  {M}-estimation.
\newblock In \textit{Advances in Neural Information Processing Systems}.

\bibitem[{Ji and Jin(2012)}]{ji2012ups}
\textsc{Ji, P.} and \textsc{Jin, J.} (2012).
\newblock {UPS} delivers optimal phase diagram in high-dimensional variable
  selection.
\newblock \textit{The Annals of Statistics} \textbf{40} 73--103.

\bibitem[{Loh and Wainwright(2015)}]{LWa15}
\textsc{Loh, P.-L.} and \textsc{Wainwright, M.~J.} (2015).
\newblock Regularized {M}-estimators with nonconvexity: Statistical and
  algorithmic theory for local optima.
\newblock \textit{The Journal of Machine Learning Research} \textbf{16}
  559--616.

\bibitem[{Loh and Wainwright(2017)}]{LWa17}
\textsc{Loh, P.-L.} and \textsc{Wainwright, M.~J.} (2017).
\newblock Support recovery without incoherence: A case for nonconvex
  regularization.
\newblock \textit{The Annals of Statistics} \textbf{45} 2455--2482.

\bibitem[{Mallows(1973)}]{Mal73}
\textsc{Mallows, C.~L.} (1973).
\newblock Some comments on $c_p$.
\newblock \textit{Technometrics} \textbf{15} 661--675.

\bibitem[{McCracken and Ng(2016)}]{MN16}
\textsc{McCracken, M.~W.} and \textsc{Ng, S.} (2016).
\newblock {FRED-MD}: A monthly database for macroeconomic research.
\newblock \textit{Journal of Business $\&$ Economic Statistics} \textbf{34}
  574--589.

\bibitem[{Needell and Tropp(2009)}]{needell2009cosamp}
\textsc{Needell, D.} and \textsc{Tropp, J.~A.} (2009).
\newblock Cosamp: Iterative signal recovery from incomplete and inaccurate
  samples.
\newblock \textit{Applied and Computational Harmonic Analysis} \textbf{26}
  301--321.

\bibitem[{Rudelson and Vershynin(2013)}]{RV13}
\textsc{Rudelson, M.} and \textsc{Vershynin, R.} (2013).
\newblock {H}anson-{W}right inequality and sub-{G}aussian concentration.
\newblock \textit{Electronic Communications in Probability} \textbf{18}.

\bibitem[{Schwarz(1978)}]{Sch78}
\textsc{Schwarz, G.} (1978).
\newblock Estimating the dimension of a model.
\newblock \textit{The Annals of Statistics} \textbf{6} 461--464.

\bibitem[{Shen et~al.(2012)Shen, Pan and Zhu}]{SPZ12}
\textsc{Shen, X.}, \textsc{Pan, W.} and \textsc{Zhu, Y.} (2012).
\newblock Likelihood-based selection and sharp parameter estimation.
\newblock \textit{Journal of the American Statistical Association} \textbf{107}
  223--232.

\bibitem[{Shen et~al.(2013)Shen, Pan, Zhu and Zhou}]{SPZ13}
\textsc{Shen, X.}, \textsc{Pan, W.}, \textsc{Zhu, Y.} and \textsc{Zhou, H.}
  (2013).
\newblock On constrained and regularized high-dimensional regression.
\newblock \textit{Annals of the Institute of Statistical Mathematics}
  \textbf{65} 807--832.

\bibitem[{Tibshirani(1996)}]{Rob96}
\textsc{Tibshirani, R.} (1996).
\newblock Regression shrinkage and selection via the lasso.
\newblock \textit{Journal of the Royal Statistical Society: Series B
  (Methodological)} \textbf{58} 267--288.

\bibitem[{van Handel(2016)}]{van16}
\textsc{van Handel, R.} (2016).
\newblock Probability in high dimension.
\newblock Lecture notes, Princeton University.

\bibitem[{Wainwright(2019)}]{Wai19}
\textsc{Wainwright, M.~J.} (2019).
\newblock \textit{High-Dimensional Statistics: A Non-Asymptotic Viewpoint}.
\newblock Cambridge University Press.

\bibitem[{Wang et~al.(2020)Wang, Weng and Maleki}]{wang2020bridge}
\textsc{Wang, S.}, \textsc{Weng, H.} and \textsc{Maleki, A.} (2020).
\newblock Which bridge estimator is the best for variable selection?
\newblock \textit{The Annals of Statistics} \textbf{48} 2791--2823.

\bibitem[{Wasserman and Roeder(2009)}]{wasserman2009high}
\textsc{Wasserman, L.} and \textsc{Roeder, K.} (2009).
\newblock High dimensional variable selection.
\newblock \textit{The Annals of Statistics} \textbf{37} 2178.

\bibitem[{Xiong(2014)}]{xiong2014better}
\textsc{Xiong, S.} (2014).
\newblock Better subset regression.
\newblock \textit{Biometrika} \textbf{101} 71--84.

\bibitem[{Zhang(2010)}]{Zha10}
\textsc{Zhang, C.-H.} (2010).
\newblock Nearly unbiased variable selection under minimax concave penalty.
\newblock \textit{The Annals of Statistics} \textbf{38} 894--942.

\bibitem[{Zhang and Zhang(2012)}]{ZZh12}
\textsc{Zhang, C.-H.} and \textsc{Zhang, T.} (2012).
\newblock A general theory of concave regularization for high-dimensional
  sparse estimation problems.
\newblock \textit{Statistical Science} \textbf{27} 576--593.

\bibitem[{Zhao and Yu(2006)}]{ZY06}
\textsc{Zhao, P.} and \textsc{Yu, B.} (2006).
\newblock On model selection consistency of {L}asso.
\newblock \textit{Journal of Machine Learning Research} \textbf{7} 2541--2563.

\bibitem[{Zou(2006)}]{Zou06}
\textsc{Zou, H.} (2006).
\newblock The adaptive {L}asso and its oracle properties.
\newblock \textit{Journal of the American Statistical Association} \textbf{101}
  1418--1429.

\bibitem[{Zou and Hastie(2005)}]{ZHa05}
\textsc{Zou, H.} and \textsc{Hastie, T.} (2005).
\newblock Regularization and variable selection via the elastic net.
\newblock \textit{Journal of the Royal Statistical Society: Series B
  (Statistical Methodology)} \textbf{67} 301--320.

\end{thebibliography}

%%%%%%%%%%%%%%%%%%%%%%%%%%%%%%%%%%%%%%%%%%%%

\newpage
\bigskip
\begin{center}
{\large\bf SUPPLEMENTARY MATERIAL}
\end{center}

\section{Proof of main theorems}

\subsection{Proof of Theorem \ref{thm-selection-consistency}}
For $t\in \{1,\cdots, s\}$, let $\cA_t :=\{\cS\subset [p]: |\cS|=s, |\cS\setminus\cS^*|=t\}$ (i.e., the set of the sets that have exactly $t$ different elements compared with $\cS^*$). Then we have $\cA(s)=\cup_{t\in[s]}\cA_t$.

Now we fix $t\in[s]$. For any $\cS\in\cA_t$, define $\cS_0:=\cS^*\setminus\cS$. Note that
\begin{equation}
\begin{aligned}
& n^{-1}(R_{\cS} - R_{\cS^*})  = n^{-1} \bigl\{\by^\top(\bI-\bP_{\bX_{\cS}})\by - \by^\top(\bI-\bP_{\bX_{\cS^*}})\by\bigr\}  \\
& = n^{-1}\bigl\{(\bX_{\cS_0}\bbeta^*_{\cS_0}+\bepsilon)^\top(\bI-\bP_{\bX_{\cS}})(\bX_{\cS_0}\bbeta^*_{\cS_0}+\bepsilon) - \bepsilon^\top(\bI-\bP_{\bX_{\cS^*}})\bepsilon\bigr\} \\
& = \bbeta^{*\top}_{\cS_0}\hat\bD(\cS)\bbeta^{*}_{\cS_0} + 2n^{-1}\bepsilon^\top(\bI-\bP_{\bX_{\cS}})\bX_{\cS_0}\bbeta^*_{\cS_0} - n^{-1} \bepsilon^\top(\bP_{\bX_{\cS}}-\bP_{\bX_{\cS^*}})\bepsilon \\
& = \eta \bbeta^{*\top}_{\cS_0}\hat\bD(\cS)\bbeta^{*}_{\cS_0} + 2^{-1}(1 - \eta)\bbeta^{*\top}_{\cS_0}\hat\bD(\cS)\bbeta^{*}_{\cS_0} + 2n^{-1}\bepsilon^\top(\bI-\bP_{\bX_{\cS}})\bX_{\cS_0}\bbeta^*_{\cS_0} \\
& \qquad + 2^{-1}(1 - \eta)\bbeta^{*\top}_{\cS_0}\hat\bD(\cS)\bbeta^{*}_{\cS_0} -  n^{-1} \bepsilon^\top(\bP_{\bX_{\cS}}-\bP_{\bX_{\cS^*}})\bepsilon.
\end{aligned}
\end{equation}
In the sequel, we show that the following two inequalities hold with high probability:
\begin{align}
& \left| 2n^{-1}\bigl\{(\bI-\bP_{\bX_{\cS}})\bX_{\cS_0}\bbeta^*_{\cS_0}\bigr\}^\top\bepsilon\right|<2^{-1}(1 -  \eta)\bbeta^{*\top}_{\cS_0}\hat\bD(\cS)\bbeta^{*}_{\cS_0}, \label{eq:thm2.1_t1}\\
& n^{-1}\bepsilon^\top(\bP_{\bX_{\cS}}-\bP_{\bX_{\cS^*}})\bepsilon<2^{-1}(1 - \eta )\bbeta^{*\top}_{\cS_0}\hat\bD(\cS)\bbeta^{*}_{\cS_0},  \label{eq:thm2.1_t2}
\end{align}
so that $n^{-1}(R_{\cS} - R_{\cS^*}) > \eta \bbeta^{*\top}_{\cS_0}\hat\bD(\cS)\bbeta^{*}_{\cS_0}$.

First, define
\[
\bgamma_{\cS}:=n^{- 1 / 2} (\bI-\bP_{\bX_{\cS}})\bX_{\cS_0}\bbeta^*_{\cS_0}.
\]
Then $\lVert{\bgamma_{\cS}}\rVert_2^2 = \bbeta^{*\top}_{\cS_0}\hat\bD(\cS)\bbeta^{*}_{\cS_0}$, and \eqref{eq:thm2.1_t1} is equivalent to
\begin{equation}
\lvert \bgamma_{\cS}^\top\bepsilon\rvert/ \lVert \bgamma_{\cS}\rVert_2 \le \frac{(1 - \eta)n^{1 / 2}}{4}\lVert \bgamma_{\cS}\rVert_2.
\end{equation}
Given that all the entries of $\bepsilon$ are i.i.d. sub-Gaussian with $\psi_2$-norm bounded by $\sigma$, applying Hoeffding's inequality yields that for any $x > 0$,
$$
\P(\lvert \bgamma_{\cS}^\top\bepsilon\rvert/ \lVert \bgamma_{\cS}\rVert_2>\sigma x)\leq 2e^{-x^2/2}.
$$
Define $\hat M_t:=\sup_{\cS\in\cA_t} \lvert \bgamma_{\cS}^\top\bepsilon\rvert/ \lVert \bgamma_{\cS}\rVert_2.$
Then a union bound over all $\cS\in \cA_t$ yields that for any $\xi > 0$,
\[
\P(\hat M_t> \xi \sigma\sqrt{t\log p})\leq 2\lvert\cA_t\rvert e^{-(\xi ^2 t\log p)/ 2}=
\begin{pmatrix}
p-s\\t
\end{pmatrix}
\begin{pmatrix}
s\\t
\end{pmatrix}
2e^{-(\xi ^ 2t\log p) / 2}
\leq 2e^{-(\xi ^2 / 2 - 2)t\log p}.
\]
%	Note that
%	\begin{equation}
%		\label{eq:gamma_s}
%		\ltwonorm{\bgamma_{\cS}} \ge \widehat \lambda_m^ {1 / 2} \ltwonorm{\bbeta^*} \ge (t\widehat \lambda_m)^{1 / 2} \mu_*.
%	\end{equation}
%	Hence, whenever
%	\[
%	\mu_* \ge \frac{4\xi \sigma}{1 - \eta}\biggl(\frac{\log p}{ n\widehat \lambda_m}\biggr)^{1 / 2},
%	\]
Therefore, whenever
\[
\frac{\inf_{\cS \in \cA_t}\|\bgamma_{\cS}\|_2}{t ^ {1 / 2}} \ge \frac{4\xi\sigma}{1 - \eta}\biggl(\frac{\log p}{n}\biggr)^{1 / 2},
\]
we have that
\[
\P\biggl(\widehat M_t > \frac{(1 - \eta)n^{1 / 2}}{4} \inf_{\cS \in \cA_t}\ltwonorm{\bgamma_{\cS}}\biggr) \le 2 e^{- (\xi ^ 2 - 2)t\log p},
\]
which implies that
\beq
\label{eq:thm2.1_t1_bound}
\P\biggl( \exists \cS \in \cA_t, \frac{\lvert \bgamma_{\cS}^\top\bepsilon\rvert}{\lVert \bgamma_{\cS}\rVert_2} > \frac{(1 - \eta)n^{1 / 2}}{4}\lVert \bgamma_{\cS}\rVert_2\biggr) \le 2e^{-(\xi^2 - 2)t\log p}.
\eeq

As for \eqref{eq:thm2.1_t2}, define
\[
\hat\delta_t:=\max_{\cS\in\cA_t}\frac{1}{n}\bepsilon^\top(\bP_{\bX_{\cS}}-\bP_{\bX_{\cS^*}})\bepsilon.
\]
%and define the events
%\[
%E_t^1:=\bigl\{ \hat M_t< (nt\hat\lambda_t^m)^{1 / 2} \min_{j\in \cS^*}\lvert \beta_j^*\rvert / 4\bigr\}~~~\text{and}~~~ E_t^2:=\bigl\{ \hat\delta_t<t \hat\lambda_t^m\min_{j\in \cS^*}\lvert \beta_j^*\rvert^2 /  2\bigr\}.
%\]
% According to the above deduction, for each $t\in\{1,\cdots, s\}$, $E_t\supseteq E_t^1\cap E_t^2$. In the following, we show that each of $E_t^1$ and $E_t^2$ holds with high probability, thus ensuring that $E$ holds with high probability.
%
%We first consider $E_t^1$. Fix $\cS\in\cA_t$.
%Let $x=C_1\sqrt{t\log p}$. Then we obtain that
%$$
%\P(\lvert \bgamma_{\cS}^\top\bepsilon\rvert/ \lVert \bgamma_{\cS}\rVert_2>C_1\sigma \sqrt{t\log p})\leq 2e^{-C_2t\log p},
%$$
%where $C_2=C_1^2/2$.  With a suitably large choice of $C_1$, we have $C_3>0$. Therefore, as long as $\min_{j\in\cS^*}\lvert\beta_j^*\rvert\geq C_4\sqrt{\sigma^2/\hat\lambda_m}\cdot\sqrt{\log p/n}$ where $C_4=4C_1$, for every $t \in [s]$,
%\begin{equation}\label{eq-E1}
%\P\left((E_1^t)^c\right)\leq 2e^{-C_3t\log p}.
%\end{equation}
%Now we analyze $E_t^2 =\bigl\{\hat\delta_t<t \hat\lambda_t^m\min_{j\in \cS}\lvert \beta_j^*\rvert^2 / 2\bigr\}$.
Fix any $\cS\in\cA_t$, let $\cU$, $\cV$ be the orthogonal complement of $\cW:=\operatorname{colspan}(\bX_{\cS^*\cap\cS})$ as a subspace of $\operatorname{colspan}(\bX_{\cS})$ and $\operatorname{colspan}(\bX_{\cS^*})$ respectively. Then $\dim(\cU)=\dim(\cV)=t$, and
\begin{align}\label{eq3}
\frac{1}{n}\bepsilon^\top(\bP_{\bX_{\cS}}-\bP_{\bX_{\cS^*}})\bepsilon
&=\frac{1}{n}\bepsilon^\top(\bP_{\cW}+\bP_{\cU})\bepsilon-\frac{1}{n}\bepsilon^\top(\bP_{\cW}+\bP_{\cV})\bepsilon\nonumber\\
&=\frac{1}{n}\bepsilon^\top(\bP_{\cU}-\bP_{\cV})\bepsilon.
\end{align}
By \cite[][Theorem~1.1]{RV13}, there exists a universal constant $c>0$ such that for any $x>0$,
$$
\P(\lvert\bepsilon^\top\bP_{\cU}\bepsilon-\E\bepsilon^\top\bP_{\cU}\bepsilon\rvert>\sigma^2x)\leq 2e^{-c\min({x^2}/{\fnorm{\bP_{\cU}}^2},\hspace{.03cm} {x}/{\ltwonorm{\bP_\cU}})}=2e^{-c\min(x^2/t,\hspace{.03cm}x)}.
$$
Similarly,
$$
\P(\lvert\bepsilon^\top\bP_{\cV}\bepsilon-\E\bepsilon^\top\bP_{\cV}\bepsilon\rvert>\sigma^2x)\leq 2e^{-c\min(x^2/t, x)}.
$$
Noticing that $\E(\bepsilon^\top\bP_{\cV}\bepsilon) = \E\tr(\bP_{\cV}\bepsilon\bepsilon^\top)=\Var(\epsilon_1)\tr(\bP_{\cV})=t\Var(\epsilon_1) = \E(\bepsilon^\top\bP_{\cU}\bepsilon)$,
we combine the above two inequalities and obtain that
$$
\P(\lvert\bepsilon^\top\bP_{\cU}\bepsilon-\bepsilon^\top\bP_{\cV}\bepsilon\rvert>2\sigma^2 x )\leq 4e^{-c\min(x^2/t, x)}.
$$
%Now let $x=C_5t\log p$ with some suitably large constant $C_5$.
Given that $\log p > 1$ and that \eqref{eq3} holds,  applying a union bound over  $\cS\in\cA_t$ yields that for any $\xi > 1$,
%\begin{align*}
%\P\biggl(\frac{1}{n}\lvert\bepsilon^\top(\bP_{\bX_{\cS}}-\bP_{\bX_{\cS^*}})\bepsilon\rvert>\frac{2C_5\sigma^2t\log p}{n}\biggr)\leq 4e^{-C_6t\log p},
%\end{align*}
\begin{align}\label{eq4}
\P\left(\hat\delta_t>\frac{2\xi \sigma^2t\log p}{n}\right)\leq  4\lvert\cA_t\rvert e^{-c \xi t\log p}=
\begin{pmatrix}
p-s\\t
\end{pmatrix}
\begin{pmatrix}
s\\t
\end{pmatrix}
4e^{- c\xi t\log p}\leq 4e^{-(c\xi - 2) t\log p},
\end{align}
Therefore, whenever
\[
\frac{\inf_{\cS \in \cA_t}\|\bgamma_{\cS}\|_2}{t ^ {1 / 2}}  \ge \biggl (\frac{4\xi\sigma ^ 2\log p}{n (1 - \eta)}\biggr)^{1 / 2},
\]
%	\[
%	\mu_* > 2\sigma \biggl\{\frac{\xi \log p}{n(1 - \eta)\widehat \lambda_m}\biggr\}^{1 / 2},
%	\]
we have that
\[
\P\biggl(\widehat \delta_t > \min_{\cS \in \cA_t} \frac{1 - \eta}{2} \|\bgamma_{\cS}\|_2^2 \biggr) \le 4e^{- (c\xi - 2)t\log p},
\]
which further implies that
\begin{equation}
\label{eq:thm2.1_t2_bound}
\P\biggl(\exists \cS \in \cA_t, \frac{1}{n} \bepsilon^\top(\bP_{\bX_\cS} - \bP_{\bX_{\cS^*}})\bepsilon \ge \frac{1 - \eta}{2}\ltwonorm{\bgamma_{\cS}}^2\biggr) \le 4e^{-(c\xi - 2)t\log p}.
\end{equation}

Finally, combining \eqref{eq:thm2.1_t1_bound} and \eqref{eq:thm2.1_t2_bound} and  applying a union bound with $t \in [s]$, we deduce that for any $\xi > \max(1, 2c^{-1})$ and $0< \eta < 1$, if
\[
\inf_{\cS \in \cA(s)}\frac{\|\bgamma_{\cS}\|_2}{t ^ {1 / 2}} \ge \frac{4\xi\sigma}{1 - \eta}\biggl(\frac{\log p}{n}\biggr)^{1 / 2},
\]
then given that $\tau_*(s) \le \bbeta^{*\top}_{\cS_0}\hat\bD(\cS)\bbeta^{*}_{\cS_0}$ for any $\cS \in \cA(s)$, we have that
\[
\P\biggl(\forall \cS \in \cA, R_{\cS} - R_{\cS^*} >  n\eta\tau_*(s) \biggr) \ge 1 - 4s\bigl\{p^{- (c\xi - 2)} + p^{-(\xi - 2)}\bigr\},
\]
as desired.
%where we write $C_6 = c\min\{C_5, C_5^2\}$ and $C_7=C_6-2$. With suitably large choice of $C_5$, we have $C_7>0$. Therefore, as long as $\min_{j\in\cS^*}\lvert\beta_j^*\rvert\geq C_8 \sqrt{\sigma^2/\hat\lambda_m}\cdot\sqrt{\log p/n}$ with $C_8=2\sqrt{C_5}$, for every $t\in \{1,2,\cdots, s\}$,
%\begin{equation}\label{eq-E2}
%\P\left((E_t^2)^c\right)\leq 4e^{-C_7t\log p}.
%\end{equation}
%
%
%Finally, we conclude from (\ref{eq-E1}) and (\ref{eq-E2}) that, as long as $\min_{j\in\cS^*}\lvert\beta_j^*\rvert\geq C_9\sqrt{\sigma^2/\hat\lambda_m}\cdot\sqrt{\log p/n}$ with $C_9=\max\{C_4, C_8\}$,
%\begin{align*}
%\P(E^c)&= \P(\cup_{t\in[s]}E_t^c)\leq \sum_{t=1}^s\P(E_t^c)\leq \sum_{t=1}^s \left(\P((E_t^1)^c)+\P((E_t^2)^c)\right)\\
%&\leq \sum_{t=1}^s\left(2e^{-C_3t\log p}+4e^{-C_7t\log p}\right)\leq 6se^{-C_{10}t\log p}.
%\end{align*}
%Here $C_{10}=\min\{C_3, C_7\}$.
%(Combining all the relationships between the constants, we actually obtain that for any $C_{10}>0$, by choosing $C_9=\max\{4\sqrt{2C_{10}+2}, \sqrt{4(C_{10}+2)/c}+1, 2[(C_{10}+2)/c]^{1/4}+1\}$, the above inequality holds with desired probability.)

\subsection{Proof of Theorem \ref{thm:lower_bound}}
For any set $\cS \in \cC_{j_0}$, note that $\widehat \bD(\cS)$ is now reduced to be a scalar, which we thus use $\widehat D(\cS)$ to denote. We have that
\begin{align}\label{eq:lower_bound_loss_comparison}
n^{-1}(R_{\cS} - R_{\cS^*}) & = n^{-1} \bigl\{\by^\top(\bI-\bP_{\bX_{\cS}})\by - \by^\top(\bI-\bP_{\bX_{\cS^*}})\by\bigr\} \nonumber \\
& = \widehat D(\cS){\beta^{*}_{j_0}}^2  + 2n^{-1}\beta^*_{j_0} \bepsilon^\top(\bI-\bP_{\bX_{\cS}})\bX_{j_0} - n^{-1} \bepsilon^\top(\bP_{\bX_{\cS}}-\bP_{\bX_{\cS^*}})\bepsilon.
\end{align}
We first provide a lower bound on $\sup_{\cS\in\cC_{j_0}}n^{-1} \bepsilon^\top(\bP_{\bX_{\cS}}-\bP_{\bX_{\cS^*}})\bepsilon$. Recall that $\cS^*_0 = \cS^* \setminus \{j_0\}$, and that for any $j \in [p] \setminus \cS^*$, $\widetilde \bu_j:=(\bI-\bP_{\bX_{\cS^*_0}})\bX_j$ and $\overline \bu_j := \widetilde \bu_j / \ltwonorm{\widetilde\bu_j}$. We have that
\begin{equation}\label{eq:lower_bound_quadratic}
\sup_{\cS\in\cC_{j_0}}\frac{1}{n} \bepsilon^\top(\bP_{\bX_{\cS}}- \bP_{\bX_{\cS^*}})\bepsilon = \sup_{j\notin\cS^*}\frac{1}{n}\bepsilon^\top(\overline\bu_j\overline\bu_j^\top-\overline\bu_{j_0}\overline\bu_{j_0}^\top)\bepsilon.
\end{equation}
We start with a lower bound of the expectation of the above term. By the Cauchy--Schwarz inequality,
\begin{align*}
\E\biggl\{\sup_{j\notin\cS^*}\frac{1}{n}\bepsilon^\top(\overline\bu_j\overline\bu_j^\top-\overline\bu_{j_0}\overline\bu_{j_0}^\top)\bepsilon\biggr\} = \E\biggl\{\sup_{j\notin\cS^*}\frac{1}{n}(\overline\bu_j^\top\bepsilon)^2\biggr\}-\frac{\sigma^2}{n} \geq \frac{1}{n}\Bigl\{\E \sup_{j\notin\cS^*}(\overline\bu_j^\top\bepsilon)\Bigr\}^2-\frac{\sigma^2}{n}.
\end{align*}
By Sudakov's lower bound on Gaussian processes,
$$
\E \sup_{j\notin\cS^*}(\overline\bu_j^\top\bepsilon)\geq \sup_{\delta > 0}\frac{\delta}{2}\bigl\{\log M(\delta, \{\overline\bu_j\}_{j\notin\cS^*})\bigr\} ^ {1 / 2}.
$$
Combining the two inequalities gives
\begin{equation}\label{eq:lower_bound_quadratic_1}
\begin{aligned}
\E \sup_{j\notin\cS^*}\frac{1}{n}\bepsilon^\top(\overline\bu_j\overline \bu_j^\top-\overline\bu_{j_0}\overline\bu_{j_0}^\top)\bepsilon & \geq   \frac{\sigma^2}{n}\biggl\{\sup_{\delta>0}\frac{\delta^2}{4}\log M(\delta, \{\overline\bu_j\}_{j\notin\cS^*}) - 1\biggr\} \\
& \ge \frac{\sigma ^ 2}{n} \biggl(\frac{\delta_0^2 c_{\delta_0}}{4}\log p - 1\biggr).
\end{aligned}
\end{equation}
Now we bound the variance of $\sup_{j\notin\cS^*}\frac{1}{n}\bepsilon^\top(\overline\bu_j\overline\bu_j^\top-\overline\bu_{j_0}\overline\bu_{j_0}^\top)\bepsilon$. We have that
\begin{equation}
\label{eq:lower_bound_quadratic_2}
\begin{aligned}
\Var\biggl\{\sup_{j\notin\cS^*}\frac{1}{n}\bepsilon^\top(\overline\bu_j\overline\bu_j^\top-\overline\bu_{j_0}\overline\bu_{j_0}^\top)\bepsilon\biggr\} & = \frac{1}{n^2}\Var\biggl\{\sup_{j\notin \cS^*}(\overline\bu_j^\top\bepsilon)^2-(\overline\bu_{j_0}^\top\bepsilon)^2\biggr\} \\
& \leq \frac{2}{n ^ 2}\biggl[\Var\biggl\{\sup_{j\notin \cS^*}(\overline\bu_j^\top\bepsilon)^2\biggr\} + \Var\{(\overline \bu_0 ^\top \bepsilon)^ 2\}\biggr] \\
& = \frac{2}{n ^ 2}\biggl[\Var\biggl\{\sup_{j\notin \cS^*}(\overline\bu_j^\top\bepsilon)^2\biggr\} + 2\sigma^4 \biggr].
\end{aligned}
\end{equation}
According to Lemma \ref{lem:var_of_max},
\begin{equation}\label{eq:lower_bound_quadratic_22}
\Var\biggl\{\sup_{j\notin \cS^*}(\overline\bu_j^\top\bepsilon)^2\biggr\}=\frac{2}{n^2}\Var\biggl\{\max\biggl(\sup_{j\notin \cS^*}\overline\bu_j^\top\bepsilon, \sup_{j
	\notin \cS^*} - \overline\bu_j^\top\bepsilon \biggr)^2\biggr\}\leq \frac{4}{n^2}\Var(Z^2),
\end{equation}
where $Z :=\sup_{j\notin \cS^*}\overline \bu_j^\top\bepsilon$. Besides,
\begin{equation}
\label{eq:var_z2}
\begin{aligned}
\Var(Z^2)& = \Var\bigl\{ (Z - \E Z) ^ 2 + 2(\E Z)Z - (\E Z) ^ 2 \bigr\} = \Var \bigl\{ (Z - \E Z) ^ 2 + 2 (\E Z)Z \bigr\}\\
& \le 2 \Var\bigl\{ (Z - \E Z) ^ 2\bigr\} + 8 (\E Z)^2 \Var(Z) \\
& = 2 \E \bigl\{ (Z - \E Z) ^ 4\bigr\} - 2\Var(Z)^2 + 8 (\E Z)^2 \Var(Z).
\end{aligned}
\end{equation}
According to Lemma \ref{lem:concentration_sup_of_gaussian}, $Z$ is $\sigma^2$-subgaussian. Hence, for any $q\geq 1$, $\left(\E |(Z - \E Z)/\sigma|^q\right)^{1/q}\lesssim \sqrt{q}$. Therefore, $\Var(Z^2)\lesssim \sigma ^ 4 + \sigma ^ 2(\E Z) ^ 2$. In addition, by \cite[][Corollary~5.25]{van16}, we have that
\[
	\E Z \le 12\sigma \int_{\delta = 0}^{\infty} \{\log N(\delta, \{\overline \bu_j\}_{j \notin \cS^*})\}^{1 / 2} d\delta,
\]
where $N(\delta, \{\overline \bu_j\}_{j \notin \cS^*})$ is the $\delta$-covering number of $\{\overline \bu_j\}_{j \notin \cS^*}$ under Euclidean distance. Given that $N(\delta, \{\overline\bu_j\}_{j \in \cS^*}) = 1$ for any $\delta > 2^{1 / 2}$, we deduce that $\E Z \lesssim \sigma(\log p) ^ {1 / 2}$. Therefore, $\Var(Z ^ 2) \lesssim \sigma ^ 4 \log p$. Combining this bound with \eqref{eq:lower_bound_quadratic_2}, \eqref{eq:lower_bound_quadratic_22} and \eqref{eq:var_z2} yields that there exists a universal constant $C_1 > 0$ such that
\[
\Var\biggl\{\sup_{j\notin\cS^*}\frac{1}{n}\bepsilon^\top(\overline\bu_j\overline\bu_j^\top-\overline\bu_{j_0}\overline\bu_{j_0}^\top)\bepsilon\biggr\}  \le \frac{C_1 \sigma^4 \log p}{n ^ 2}.
\]
Finally, by Markov's inequality, for any $t > 0$, we have that
\[
\PP \biggl\{\sup_{\cS\in\cA_1}\frac{1}{n} \bepsilon^\top(\bP_{\bX_{\cS}}- \bP_{\bX_{\cS^*}})\bepsilon \le  \frac{\sigma ^ 2}{n} \biggl(\frac{\delta_0^2 c_{\delta_0}}{4}\log p - 1\biggr) - \frac{t\sigma ^ 2 (C_1 \log p) ^ {1 / 2}}{n}\biggr\} \le t^{-2},
\]
from which we further deduce that if $\log p > 5 / (\delta_0^2 c_{\delta_0})$, then there exists $C_2(\delta_0) > 0$ such that
\begin{equation}
\label{eq:lower_bound_quadratic_final}
\PP \biggl\{\sup_{\cS\in\cA_1}\frac{1}{n} \bepsilon^\top(\bP_{\bX_{\cS}}- \bP_{\bX_{\cS^*}})\bepsilon \le \frac{\sigma^2 \delta_0^2 c_{\delta_0} \log p}{21n} \biggr\}	\le \frac{C_2(\delta_0)}{\log p}.
\end{equation}

Now we proceed to give an upper bound of the second term on the right hand side of (\ref{eq:lower_bound_loss_comparison}), i.e., $2n^{-1}\beta^*_{j_0}\bepsilon^\top(\bI-\bP_{\bX_{\cS}})\bX_{j_0}$, for all $\cS\in{\cA_1}(j_0)$. Recall that we have defined $\bgamma_\cS=\frac{1}{\sqrt{n}}(\bI-\bP_{\bX_\cS})\bX_{\cS_0}\bbeta_{\cS_0}^*$, and that  $\ltwonorm{\bgamma_{\cS}}^2=(\bbeta_{\cS_0}^*)^\top\widehat\bD(\cS)\bbeta_{\cS_0}^*$, where $\cS_0=\cS^*\setminus\cS$. By definition,
$$
%\max_{\cS\in\cC_{j_0}}\ltwonorm{\bgamma_\cS}\leq \sigma\biggl(\frac{c\log p}{n}\biggr) ^{1 / 2}.
\sup_{\cS\in\cC_{j_0}}\ltwonorm{\bgamma_\cS} ^ 2\leq \tau^*.
$$
On the other hand, a union bound yields that for any $\xi>2^{1 / 2}$,
$$
\P\biggl\{\sup_{\cS\in\cC_{j_0}}\frac{|\bgamma_\cS^\top\bepsilon|}{\ltwonorm{\bgamma_\cS}}\geq \xi\sigma(\log p) ^ {1 / 2} \biggr \}\leq 2e^{-(\xi^2/2-1)\log p}.
$$
Let $\xi=2$. Then the two inequalities above yield that
\begin{equation}\label{eq:lower_bound_linear_final}
\P\biggl\{\sup_{\cS\in\cC_{j_0}}\left|2n^{-1}\beta^*_{j_0}\bepsilon^\top(\bI-\bP_{\bX_{\cS}})\bX_{j_0}\right|\geq 2\sigma \biggl(\frac{\tau^* \log p}{n} \biggr)^{1 / 2} \biggr\}\leq \frac{2}{p}.
\end{equation}

Finally, combining (\ref{eq:lower_bound_loss_comparison}), (\ref{eq:lower_bound_quadratic_final}) and (\ref{eq:lower_bound_linear_final}), we obtain that with probability at least $1-2p^{-1} - C_2(\delta_0)(\log p) ^ {-1}$,
\begin{align}
\inf_{\cS\in\cC_{j_0}}n^{-1}(R_{\cS} - R_{\cS^*})
& = \inf_{\cS\in\cC_{j_0}} \bigl\{\widehat D(\cS){\beta^{*}_{j_0}} ^ 2 + 2n^{-1}\beta^*_{j_0} \bepsilon^\top(\bI-\bP_{\bX_{\cS}})\bX_{j_0}  - n^{-1} \bepsilon^\top(\bP_{\bX_{\cS}}-\bP_{\bX_{\cS^*}})\bepsilon \bigr\}\nonumber\\
& \le  \tau^* + \sup_{\cS\in\cA_1}\biggl|2n^{-1}\beta^*_{j_0} \bepsilon^\top(\bI-\bP_{\bX_{\cS}})\bX_{j_0} \biggr| - \sup_{\cS \in \cA_1} \frac{1}{n} \bepsilon^\top (\bP_{\bX_{\cS}} - \bP_{\bX_{\cS^*}}) \bepsilon \nonumber\\
& \le \tau ^ * + 2\sigma \biggl(\frac{\tau^* \log p}{n} \biggr)^{1 / 2} - \frac{\sigma^2 \delta_0^2 c_{\delta_0} \log p}{21n}. \nonumber
\end{align}
The conclusion thus follows by our condition on $\tau ^ *$.

\begin{lem}\label{lem:normalize}
	Suppose $\bu_1, \bu_2\in \R^d$ such that $0 < \ltwonorm{\bu_1},\ltwonorm{\bu_2}\leq 1$. Define $\bar \bu_i = \bu_i / \ltwonorm{\bu_i}$ for $i = 1, 2$. Then
	$$
	\ltwonorm{\bu_1-\bu_2}\geq \min\{\ltwonorm{\bu_1}, \ltwonorm{\bu_2}\} \ltwonorm{\bar\bu_1-\bar\bu_2}.
	$$
\end{lem}
\begin{proof}
	Consider a Euclidean space where $\bu_1=\vv{OA}$, $\bu_2=\vv{OB}$, $\bar\bu_1=\vv{O\bar{A}}$, $\bar\bu_2=\vv{O\bar{B}}$. Without loss of generality, assume that $\ltwonorm{\bu_1}\leq \ltwonorm{\bu_2}$. Let $\bu_2' = \frac{\ltwonorm{\bu_1}}{\ltwonorm{\bu_2}}\bu_2=\vv{OB_1}$. Then $|AB_1|=\ltwonorm{\bu_1}  \ltwonorm{\bar\bu_1-\bar\bu_2}$, and $|AB|=\ltwonorm{\bu_1-\bu_2}$. On the other hand, $\ltwonorm{\bu_2}=\ltwonorm{\bu_1}$, meaning that $|OA|=|OB_1|$. Thus $ABB_1$ is an obtuse triangle, and we have $|AB_1|\leq|AB|$.
\end{proof}
\begin{lem}\label{lem:var_of_max}
	Given two random variables $X_1$ and $X_2$ valued in $\RR$, $\Var\{\max(X_1, X_2)\} \leq \Var(X_1)+\Var(X_2)$.
\end{lem}
\begin{proof}
	$\Var\{\max(X_1, X_2)\} =\Var\{(X_1+X_2)/2+|X_1-X_2|/2 \}\leq \frac{1}{2}\Var(X_1+X_2)+\frac{1}{2}\Var(X_1-X_2)=\Var(X_1)+\Var(X_2)$.
\end{proof}
\begin{lem}[{\citet[][Lemma~6.12]{van16}}]
	\label{lem:concentration_sup_of_gaussian}
	Let $\{X_t\}_{t\in T}$ be a separable Gaussian process. Then $\sup_{t\in T}X_t$ is $\sup_{t\in T}\Var(X_t)$-subgaussian.
\end{lem}

\subsection{Proof of Theorem \ref{thm:power}}
For $t\in \{1, \ldots, s\}$, let $\hat\cA_t :=\{\cS\subset [p]: |\cS|=\hat s, |\cS^*\setminus\cS|=t\}$ (i.e., the set of the sets that have missed $t$ different elements in $\cS^*$). Then we have $\cA(\hat s )=\cup_{t\in[s]}\widehat\cA_t$. Our goal is to prove that with high probability, $R_{\cS^*}\leq R_{\cS} - n\eta\tau_*$ for all $\cS\in\cup_{t \ge  \delta s} \widehat\cA_t$ under condition \eqref{eq:min_tau_condition_1}, so that $\tpr(\cS) \ge 1 - \delta$ for any $\cS$ satisfiying that $|\cS| = \hat s$ and that $R_{\cS} \le R_{\cS  ^ *} + n\eta \tau_*$.

Now we fix $t\in[s]$. For any $\cS\in\widehat\cA_t$, define $\cS_0:=\cS^*\setminus\cS$.  Note that
\begin{equation}
\begin{aligned}\label{eq:loss_comparison}
n^{-1}(R_{\cS} - R_{\cS^*}) & = n^{-1} \bigl\{\by^\top(\bI-\bP_{\bX_{\cS}})\by - \by^\top(\bI-\bP_{\bX_{\cS^*}})\by\bigr\}  \\
& = n^{-1}\bigl\{(\bX_{{\cS}_0}\bbeta^*_{{\cS}_0}+\bepsilon)^\top(\bI-\bP_{\bX_{\cS}})(\bX_{{\cS}_0}\bbeta^*_{{\cS}_0}+\bepsilon) - \bepsilon^\top(\bI-\bP_{\bX_{\cS^*}})\bepsilon\bigr\} \\
& = \bbeta^{*\top}_{{\cS}_0}\hat\bD(\cS)\bbeta^{*}_{{\cS}_0} + 2n^{-1}\bepsilon^\top(\bI-\bP_{\bX_{\cS}})\bX_{{\cS}_0}\bbeta^*_{{\cS}_0} - n^{-1} \bepsilon^\top(\bP_{\bX_{\cS}}-\bP_{\bX_{\cS^*}})\bepsilon \\
& \ge \eta\tau_*(\hat s, \delta) + 2^{-1}(1 - \eta)\bbeta^{*\top}_{\cS_0}\hat\bD(\cS)\bbeta^{*}_{{\cS}_0} + 2n^{-1}\bepsilon^\top(\bI-\bP_{\bX_{\cS}})\bX_{{\cS}_0}\bbeta^*_{{\cS}_0} \\
& \qquad + 2^{-1}(1 - \eta)\bbeta^{*\top}_{{\cS}_0}\hat\bD(\cS)\bbeta^{*}_{{\cS}_0} -  n^{-1} \bepsilon^\top(\bP_{\bX_{\cS}}-\bP_{\bX_{\cS^*}})\bepsilon.
\end{aligned}
\end{equation}
In the sequel, we show that the following two inequalities hold with high probability for $ t\in[\delta s, s]$:
\begin{align}
& \left|2n^{-1}\bigl\{(\bI-\bP_{\bX_{\cS}})\bX_{{\cS}_0}\bbeta^*_{{\cS}_0}\bigr\}^\top\bepsilon\right|<2^{-1}(1 -  \eta)\bbeta^{*\top}_{{\cS}_0}\hat\bD(\cS)\bbeta^{*}_{{\cS}_0}, \label{eq:thm3.1_t1}\\
& n^{-1}\bepsilon^\top(\bP_{\bX_{\cS}}-\bP_{\bX_{\cS^*}})\bepsilon<2^{-1}(1 - \eta )\bbeta^{*\top}_{{\cS}_0}\hat\bD(\cS)\bbeta^{*}_{{\cS}_0}.  \label{eq:thm3.1_t2}
\end{align}

First, define
\[
\bgamma_{\cS}:=n^{- 1 / 2} (\bI-\bP_{\bX_{\cS}})\bX_{{\cS}_0}\bbeta^*_{{\cS}_0}.
\]
Then $\lVert{\bgamma_{\cS}}\rVert_2^2 = \bbeta^{*\top}_{{\cS}_0}\hat\bD(\cS)\bbeta^{*}_{{\cS}_0}$, and \eqref{eq:thm3.1_t1} is equivalent to
\begin{equation}
\lvert \bgamma_{\cS}^\top\bepsilon\rvert/ \lVert \bgamma_{\cS}\rVert_2 \le \frac{(1 - \eta)n^{1 / 2}}{4}\lVert \bgamma_{\cS}\rVert_2.
\end{equation}
Given that all the entries of $\bepsilon$ are i.i.d. sub-Gaussian with $\psi_2$-norm bounded by $\sigma$, applying Hoeffding's inequality yields that for any $x > 0$,
$$
\P(\lvert \bgamma_{\cS}^\top\bepsilon\rvert/ \lVert \bgamma_{\cS}\rVert_2>\sigma x)\leq 2e^{-x^2/2}.
$$
Define $\hat M_t:=\sup_{\cS\in\widehat\cA_t} \lvert \bgamma_{\cS}^\top\bepsilon\rvert/ \lVert \bgamma_{\cS}\rVert_2.$
Then a union bound over all $\cS\in \widehat\cA_t$ yields that for any $\xi > 0$,
\begin{align*}
\P(\hat M_t> \xi \sigma & \sqrt{(\hat s-s+t)\log p}) \leq 2\lvert\widehat\cA_t\rvert e^{-\{\xi ^2 (\hat s-s+t)\log p\}/ 2}\\
&=\begin{pmatrix}
p-s\\ \hat s-s+ t
\end{pmatrix}
\begin{pmatrix}
s\\t
\end{pmatrix}
\cdot 2e^{-\xi ^ 2(\hat s - s +t)\log p / 2}
\leq 2e^{-(\xi ^2 / 2 - 2)(\hat s-s+t)\log p}.
\end{align*}
Now under condition \eqref{eq:min_tau_condition_1},
%when $t\geq \delta s$, for all $\cS\in\widehat{A}_t$,
%\begin{equation}\label{eq:tau_bound_1}
%\bbeta^{*\top}_{{\cS}_0}\hat\bD(\cS)\bbeta^{*}_{{\cS}_0}\geq \tau_* t \geq  \frac{t}{\delta} \biggl(\frac{4\xi }{1 - \eta}\biggr)^2\cdot\frac{\hat s}{s}\cdot\frac{\sigma^2\log p}{n}\geq \biggl(\frac{4\xi }{1 - \eta}\biggr)^2\hat s\frac{\sigma^2\log p}{n}\geq \biggl(\frac{4\xi }{1 - \eta}\biggr)^2(\hat s-s+t)\frac{\sigma^2\log p}{n}.
%\end{equation}
we have that for any $t \in (\delta s, s]$ and any $\cS\in\widehat{A}_t$,
\begin{equation}\label{eq:tau_bound_2}
\bbeta^{*\top}_{{\cS}_0}\hat\bD(\cS)\bbeta^{*}_{{\cS}_0}\geq  \biggl(\frac{4\xi }{1 - \eta}\biggr)^2(\hat s-s+t)\frac{\sigma^2\log p}{n}.
\end{equation}
Combining the fact that $ \ltwonorm{\bgamma_{\cS}} =\sqrt{\bbeta^{*\top}_{{\cS}_0}\hat\bD(\cS)\bbeta^{*}_{{\cS}_0}}$, we obtain that
\[
\P\biggl(\widehat M_t > \frac{(1 - \eta)n^{1 / 2}}{4} \inf_{\cS \in \widehat\cA_t}\ltwonorm{\bgamma_{\cS}}\biggr) \le 2 e^{- (\xi ^ 2 / 2 - 2)(\hat s-s+t)\log p},
\]
holds under conditions in $(i)$ and $(ii)$ for different ranges of $t$. This implies that
\beq
\label{eq:t1_bound}
\P\biggl( \exists \cS \in \widehat\cA_t, \frac{\lvert \bgamma_{\cS}^\top\bepsilon\rvert}{\lVert \bgamma_{\cS}\rVert_2} > \frac{(1 - \eta)n^{1 / 2}}{4}\lVert \bgamma_{\cS}\rVert_2\biggr) \le 2e^{-(\xi^2 / 2 - 2)(\hat s-s+t)\log p}.
\eeq

As for \eqref{eq:thm3.1_t2}, define
\[
\hat\delta_t:=\max_{\cS\in\widehat\cA_t}\frac{1}{n}\bepsilon^\top(\bP_{\bX_{\cS}}-\bP_{\bX_{\cS^*}})\bepsilon.
\]
Fix any $\cS\in\widehat\cA_t$, let $\cU$, $\cV$ be the orthogonal complement of $\cW:=\operatorname{colspan}(\bX_{\cS^*\cap\cS})$ as a subspace of $\operatorname{colspan}(\bX_{\cS})$ and $\operatorname{colspan}(\bX_{\cS^*})$ respectively. Then $\dim(\cU)\le \hat s-s+t$, $\dim(\cV)\le t$, and
\begin{align}\label{eq3}
\frac{1}{n}\bepsilon^\top(\bP_{\bX_{\cS}}-\bP_{\bX_{\cS^*}})\bepsilon
&=\frac{1}{n}\bepsilon^\top(\bP_{\cW}+\bP_{\cU})\bepsilon-\frac{1}{n}\bepsilon^\top(\bP_{\cW}+\bP_{\cV})\bepsilon\nonumber\\
&=\frac{1}{n}\bepsilon^\top(\bP_{\cU}-\bP_{\cV})\bepsilon.
\end{align}
%Note that
%\[
%	\begin{aligned}
%	\sup_{\mathrm{dim}(\cU) = t} \lvert \bepsilon^\top \bP_{\cU} \bepsilon  - \E (\bepsilon^\top \bP_{\cU} \bepsilon)\rvert & \le \sup_{\mathrm{dim}(\cU) = t} \biggl\lvert \sum_{j = 1}^t \bigl\{(\bepsilon^\top \bu_j)^2 - \sigma ^ 2\bigr\}\biggr \rvert \le \sup_{\mathrm{dim}(\cU) = t}  \sum_{j = 1}^t \bigl\lvert(\bepsilon^\top \bu_j)^2 - \sigma ^ 2\bigr\rvert \\
%	& \le t(\|\bepsilon\|_2^2)
%	\end{aligned}
%\]
By \cite[][Theorem~1.1]{RV13}, there exists a universal constant $c>0$ such that for any $x>0$,
$$
\P(\lvert\bepsilon^\top\bP_{\cU}\bepsilon-\E\bepsilon^\top\bP_{\cU}\bepsilon\rvert>\sigma^2x)\leq 2e^{-c\min({x^2}/{\fnorm{\bP_{\cU}}^2},\hspace{.03cm} {x}/{\ltwonorm{\bP_\cU}})}\le 2e^{-c\min(x^2/(\hat s-s+t),\hspace{.03cm}x)}.
$$
Similarly,
$$
\P(\lvert\bepsilon^\top\bP_{\cV}\bepsilon-\E\bepsilon^\top\bP_{\cV}\bepsilon\rvert>\sigma^2x)\leq 2e^{-c\min(x^2/t, x)}.
$$
Noticing that $\E(\bepsilon^\top\bP_{\cU}\bepsilon) = \E\tr(\bP_{\cU}\bepsilon\bepsilon^\top)=\Var(\epsilon_1)\tr(\bP_{\cU})=(\hat s-s+t)\sigma^2$, and similarly $\E(\bepsilon^\top\bP_{\cU}\bepsilon) = t\sigma^2$,
we combine the above two inequalities and obtain
\begin{align*}
\P(\lvert\bepsilon^\top\bP_{\cU}\bepsilon-\bepsilon^\top\bP_{\cV}\bepsilon\rvert>(\hat s -s)\sigma^2+2x\sigma^2)&\leq \P(\lvert\bepsilon^\top\bP_{\cU}\bepsilon-\E \bepsilon^\top\bP_{\cU}\bepsilon\rvert>x\sigma^2) \\
&+\P(\lvert\bepsilon^\top\bP_{\cV}\bepsilon-\E \bepsilon^\top\bP_{\cV}\bepsilon\rvert>x\sigma^2)\leq 4e^{-c\min(x^2/(\hat s-s+t),\hspace{.03cm}x)}.
\end{align*}
%Now let $x=C_5t\log p$ with some suitably large constant $C_5$.
Given that $\log p > 1$ and that \eqref{eq3} holds,  applying a union bound over  $\cS\in\widehat\cA_t$ yields that for any $\xi > 1$, by taking $x = \xi(\hat s-s+t)$,
\begin{align}\label{eq4}
& \P\left(\hat\delta_t >\frac{3\xi \sigma^2(\hat s-s+t)\log p}{n}\right) \\
&\leq
\P(\lvert\bepsilon^\top\bP_{\cU}\bepsilon-\bepsilon^\top\bP_{\cV}\bepsilon\rvert>(\hat s -s)\sigma^2+2\sigma^2\cdot \xi(\hat s-s+t))\nonumber\\
&\leq 4\lvert\widehat\cA_t\rvert e^{-c \xi (\hat s-s+t)\log p}=
\begin{pmatrix}
p-s\\ \hat s-s+t
\end{pmatrix}
\begin{pmatrix}
s\\t
\end{pmatrix}
\cdot 4e^{- c\xi (\hat s-s+t)\log p}\nonumber\\
&\leq 4e^{-(c\xi - 2) (\hat s-s+t)\log p}.
\end{align}
Given (\ref{eq:tau_bound_2}), \eqref{eq:min_tau_condition_1} and that $\ltwonorm{\bgamma_{\cS}} =\sqrt{\bbeta^{*\top}_{{\cS}_0}\hat\bD(\cS)\bbeta^{*}_{{\cS}_0}}$, we have that
\[
\P\biggl(\widehat \delta_t > \min_{\cS \in \widehat\cA_t} \frac{1 - \eta}{2} \|\bgamma_{\cS}\|_2^2 \biggr) \le 4e^{- (c\xi - 2)(\hat s-s+t)\log p}
\]
holds for all $t\geq \delta s$.
This further implies that
\begin{equation}
\label{eq:t2_bound}
\P\biggl(\exists \cS \in \widehat\cA_t, \frac{1}{n} \bepsilon^\top(\bP_{\bX_{\cS}} - \bP_{\bX_{\cS^*}})\bepsilon \ge \frac{1 - \eta}{2}\ltwonorm{\bgamma_{\cS}}^2\biggr) \le 4e^{-(c\xi - 2)(\hat s-s+t)\log p}.
\end{equation}

To reach the final conclusion, we combine \eqref{eq:t1_bound} and \eqref{eq:t2_bound}, and apply a union bound with $t\in[\delta s, s]\cup\mathbb N$. We deduce that for any $\xi > \max(1, 2c^{-1})$ and $0< \eta < 1$, if \eqref{eq:min_tau_condition_1} holds,  we have that with probability at least $1 - 4s\bigl\{p^{- (c\xi - 2)} + p^{-(\xi - 2)}\bigr\}$, for any $\cS\in\cup_{t\geq \delta s}\widehat\cA_t$,
\[
R_{\cS} - R_{\cS^*} >  n\eta\tau_*.
\]
%On the other hand, under conditions of $(ii)$, we also combine \eqref{eq:t1_bound} and \eqref{eq:t2_bound}, while applying a union bound with $t\in\{1, \cdots, s\}$. We can see that for any $\xi > \max(1, 2c^{-1})$ and $0< \eta < 1$, if \eqref{eq:min_tau_condition_2} holds,  we have that with probability at least $1 - 4s\bigl\{p^{- (c\xi - 2)} + p^{-(\xi - 2)}\bigr\}$, for any $\cS\in\cup_{t}\widehat\cA_t$,
%\[
%R_{\cS} - R_{\cS^*} >  n\eta\tau_*.
%\]

\subsection{Proof of Theorem \ref{thm:iht}}

By Proposition \ref{prop:iht_opt}, there exists a universal constant $C_1$, such that when $t \ge C_1\kappa \log\{\cL(\widehat \bbeta^{\iht} _ 0) \allowbreak/ (n \eta \tau_*(\pi, \delta))\}$,
\beq
\label{eq:rss_betat}
\begin{aligned}
	R_{\supp(\hat\bbeta^{\iht}_t)} = \cL(\hat\bbeta^{\iht}_t) \le \min_{\bbeta \in \RR^p, \|\bbeta\|_0 \le s} \cL(\bbeta)+ n\eta \tau_*(\pi, \delta) & = \min_{\cS' \subset [p], |\cS'| = s} R_{\cS'} + n\eta \tau_*(\pi, \delta)  \\
	& \le R_{\cS ^ *} + n\eta \tau_*(\pi, \delta).
\end{aligned}
\eeq
Then the conclusion follows immediately by applying Theorem \ref{thm:power} with $\hat s = \pi$.
%\begin{cor}
%\label{cor:power}		
%	\begin{itemize}
%		\item[i)] Under all conditions of Thm. \ref{thm:power} (i), with probability at least $1 - 8sp^{- (C^{-1}\xi - 1)}$, for any $ \cS$ such that $|supp( \cS)|=\widehat s$ and $R_{\cS} \le \min_{\cS' \subset [p], |\cS'| = s} R_{\cS'} + n\eta\tau_*$, we have $TPR(\cS)\geq 1-\delta$.
%		\item[ii)] Under all conditions of Thm. \ref{thm:power} (ii), with probability at least $1 - 8sp^{- (C^{-1}\xi - 1)}$, for any $ \cS$ such that $|supp( \cS)|=\widehat s$ and $R_{\cS} \le \min_{\cS' \subset [p], |\cS'| =s} R_{\cS'} + n\eta\tilde\tau_*$, we have $TPR(\cS)=1$.
%	\end{itemize}
%\end{cor}
%To see why Corollary \ref{cor:power} is true, we only need to notice that the set
%$$
%\{\cS: R_{\cS} \le \min_{\cS' \subset [p], |\cS'| = s} R_{\cS'} + n\eta\tau_*\}
%$$
%is a subset of the set
%$$
%\{\cS: R_{\cS} \le R_{\cS^*} + n\eta\tau_*\}.
%$$
%Therefore, under the same conditions, any $\cS$ satisfying the assumptions of Theorem \ref{thm:power} also satisfies the assumptions of Corollary \ref{cor:power}.

\subsection{Proof of Corollary \ref{cor:two_stage}}

Theorem \ref{thm:iht} shows that there exist universal constants $C_1, C_2$ such that as long as $l\geq s$ and $\pi\geq 4\kappa^2l$, for any $\xi > C_1$ and $0\le  \eta < 1$,   whenever $\tau_* (\pi, \delta) \ge 16\xi ^ 2\sigma ^ 2 \log (p) / \{(1 - \eta) ^ 2n\}$ for some $\delta < s ^ {-1}$, we have with probability at least $1 - 8sp^{- (C_2^{-1}\xi - 1)}$ that $\tpr(\widehat\bbeta^{\iht}_t)=1$, or in other words, $\cS^* \subset \supp(\widehat \bbeta ^ {\iht}_ t)$, for any $t\geq C_1\kappa\log\frac{\cL(\widehat \bbeta^{\iht}_0)}{n\eta\tau_*(\pi, \delta)}$. Combining this with the definition of $\widetilde\bbeta^{\iht}_t(s)$ yields that $R_{\supp(\widetilde \bbeta^{\iht}_t(s))} \le R_{\cS^*}$.

Now we wish to apply Theorem \ref{thm-selection-consistency} to deduce the conclusion. Note that for any $\cS_1, \cS_2 \subset [p]$ such that $\cS_1 \subset \cS_2$, 
\[
	\frac{\bbeta ^{* \top}_{\cS_1 \backslash \cS ^ *}\widehat\bD(\cS_1)\bbeta ^ *_{\cS_1 \backslash \cS ^ *}}{|\cS_1 \backslash \cS ^ *|} = \frac{\ltwonorm{(\bI_n - \bP_{\bX_{\cS_1}})\bX\bbeta ^ *} ^ 2}{|\cS_1 \backslash \cS ^ *|} \ge \frac{\ltwonorm{(\bI_n - \bP_{\bX_{\cS_2}})\bX\bbeta ^ *} ^ 2}{|\cS_2 \backslash \cS ^ *|} = \frac{\bbeta ^{* \top}_{\cS_2 \backslash \cS ^ *}\widehat\bD(\cS_2)\bbeta ^ *_{\cS_2 \backslash \cS ^ *}}{|\cS_2 \backslash \cS ^ *|}. 
\]
Therefore, $\tau_*(s) \ge \tau_*(\pi) \ge \tau_*(\pi, \delta)$. The conclusion then follows by applying Theorem \ref{thm-selection-consistency}. 

{
\subsection{Brief explanation on the selected variables in the Monthly Macroeconomic Dataset}\label{section:macro-var}
In the unemployment rate association study, the variables selected by IHT are: \texttt{HWIURATIO} (Ratio of Help Wanted/No. Unemployed), \texttt{HWI} (Help-Wanted Index for United States), \texttt{DMANEMP} (All Employees: Durable goods) and \texttt{PAYEMS} (All Employees: Total nonfarm). Comparatively, LASSO puts more weight on \texttt{COMPAPFFx} (3-Month Commercial Paper Minus FEDFUNDS), \texttt{M1SL} (M1 Money Stock) and \texttt{UEMPMEAN} (Average Duration of Unemployment), while SCAD also puts more weight on \texttt{DMANEMP} (All Employees: Durable goods). In the Consumer Price Index association study, the related variables selected by IHT are \texttt{DNDGRG3M086SBEA} (Personal Cons. Exp: Nondurable goods) and \texttt{PCEPI} (Personal Cons. Expend.: Chain Index). On the other hand, LASSO and SCAD are also selecting features such as \texttt{FEDFUNDS} (Effective Federal Funds Rate), \texttt{NDMANEMP} (All Employees: Nondurable goods), \texttt{WPSID61} (PPI: Intermediate Materials), \texttt{BUSINVx} (Total Business Inventories), etc. SIS includes variables such as \texttt{WPSID61} (PPI: Intermediate Materials), \texttt{WPSID62} (PPI: Crude Materials) and \texttt{WPSFD49207} (PPI: Finished Goods) which are in the same sector.

}
%\begin{description}

%\item[Supplementary material:] The file ``supplementary.pdf" contains more details and proofs of the results in this paper.

%\end{description}

\end{document}